\documentclass[aps,prl,twocolumn,superscriptaddress,notitlepage,noshowkeys]{revtex4-2}
\usepackage[T1]{fontenc}
\usepackage[english]{babel}
\usepackage{amsmath}
\usepackage{mathtools}
\usepackage{amssymb}
\usepackage{amsthm}
\usepackage{revsymb}
\usepackage{todonotes}
\usepackage[shortlabels]{enumitem}
\usepackage{graphicx}
\usepackage{mathrsfs}
\usepackage{float}
\usepackage{comment}
\usepackage[newcommands]{ragged2e}
\makeatletter
\let\MYcaption\@makecaption
\makeatother

\usepackage[font=footnotesize]{subcaption}

\makeatletter
\let\@makecaption\MYcaption
\makeatother
\usepackage{xcolor}
\usepackage{mdframed}
\mdfsetup{
	linecolor=white,
	backgroundcolor=gray!20,
}
\usepackage{tikz}
\usetikzlibrary{positioning}
\usetikzlibrary {arrows.meta}
\usepackage{pgfplots}
\usepackage{braket}
\usepackage{bm}
\usepackage[colorlinks=true, allcolors=darkblue,breaklinks=true]{hyperref}

\usepackage{scalerel}
\usetikzlibrary{svg.path}
\definecolor{orcidlogocol}{HTML}{A6CE39}
\tikzset{
	orcidlogo/.pic={
		\fill[orcidlogocol] svg{M256,128c0,70.7-57.3,128-128,128C57.3,256,0,198.7,0,128C0,57.3,57.3,0,128,0C198.7,0,256,57.3,256,128z};
		\fill[white] svg{M86.3,186.2H70.9V79.1h15.4v48.4V186.2z}
		svg{M108.9,79.1h41.6c39.6,0,57,28.3,57,53.6c0,27.5-21.5,53.6-56.8,53.6h-41.8V79.1z M124.3,172.4h24.5c34.9,0,42.9-26.5,42.9-39.7c0-21.5-13.7-39.7-43.7-39.7h-23.7V172.4z}
		svg{M88.7,56.8c0,5.5-4.5,10.1-10.1,10.1c-5.6,0-10.1-4.6-10.1-10.1c0-5.6,4.5-10.1,10.1-10.1C84.2,46.7,88.7,51.3,88.7,56.8z};
	}
}
\newcommand{\orcidlink}[1]{\href{https://orcid.org/#1}{\mbox{\scalerel*{
				\begin{tikzpicture}[yscale=-1,transform shape]
					\pic{orcidlogo};
				\end{tikzpicture}
			}{X}}}}

\definecolor{darkblue}{rgb}{0.1, 0.1, 0.6}
\definecolor{plotblue}{RGB}{0, 114, 178}
\definecolor{plotorange}{RGB}{255, 127, 14}
\definecolor{plotgreen}{RGB}{44, 160, 44}
\definecolor{plotred}{RGB}{214, 39, 40}
\definecolor{plotpurple}{RGB}{148, 103, 189}
\definecolor{plotbrown}{RGB}{140, 86, 75}
\definecolor{plotpink}{RGB}{227, 119, 194}

\theoremstyle{definition}
\newtheorem{thm}{Theorem}
\newtheorem{lem}[thm]{Lemma}

\newtheorem*{rem*}{Remark}
\newtheorem*{ex*}{Example}
\newtheorem*{defin*}{Definition}
\newtheorem*{thm*}{Theorem}
\newtheorem*{lem*}{Lemma}
\newtheorem*{cor*}{Corollary}
\newtheorem*{def*}{Definition}
\newtheorem*{main_one}{State-Dependent Trotter Error}
\newtheorem*{main_two}{Norm Trotter Error}

\newcommand{\rme}{\mathrm{e}}
\newcommand{\rmi}{\mathrm{i}}
\newcommand{\rmd}{\mathrm{d}}
\newcommand{\tr}{\operatorname{tr}}

\newcommand{\ad}{\operatorname{ad}}
\newcommand{\Ad}{\operatorname{Ad}}
\newcommand{\id}{\mathrm{id}}
\newcommand{\op}{\mathrm{op}}
\newcommand{\F}{\mathrm{F}}

\begin{document}
\title{Lower Bounds for the Trotter Error}
\author{Alexander Hahn\orcidlink{0000-0002-4152-9854}}
\affiliation{Center for Engineered Quantum Systems, Macquarie University, 2109 NSW, Australia}
\author{Paul Hartung}
\affiliation{Department Physik, Friedrich-Alexander-Universit\"at Erlangen-N\"urnberg, Staudtstra\ss e 7, 91058 Erlangen, Germany}
\author{Daniel Burgarth\orcidlink{0000-0003-4063-1264}}
\affiliation{Department Physik, Friedrich-Alexander-Universit\"at Erlangen-N\"urnberg, Staudtstra\ss e 7, 91058 Erlangen, Germany}
\affiliation{Center for Engineered Quantum Systems, Macquarie University, 2109 NSW, Australia}
\author{Paolo Facchi\orcidlink{0000-0001-9152-6515}}
\affiliation{Dipartimento di Fisica, Universit\`a di Bari, I-70126 Bari, Italy}
\affiliation{INFN, Sezione di Bari, I-70126 Bari, Italy}
\author{Kazuya Yuasa\orcidlink{0000-0001-5314-2780}}
\affiliation{Department of Physics, Waseda University, Tokyo 169-8555, Japan}
\date{\today}

\begin{abstract}
		In analog and digital simulations of practically relevant quantum systems, the target dynamics can only be implemented approximately.
		The Trotter product formula is the most common approximation scheme as it is a generic method which allows tuning accuracy.
		The Trotter simulation precision will always be inexact for non-commuting operators, but it is currently unknown what the \emph{minimum} possible error is.
		This is an important quantity because upper bounds for the Trotter error are known to often be vast overestimates.
		Here, we present explicit lower bounds on the error, in norm and on states, allowing to derive minimum resource requirements.
		Numerical comparison with the true error shows that our bounds offer accurate and tight estimates.
\end{abstract}
\maketitle

\emph{Introduction.}---Finding approximate solutions to the Schr{\"o}dinger equation by simulating the time evolution under a given Hamiltonian is one of the most important tasks in quantum mechanics.
To tackle this problem, many simulation algorithms have been developed, which can be broadly categorized into three classes: (i) analog quantum simulation~\cite{Langford2017,MacDonell2021,Hangleiter2022,Daley2023}, (ii) digital quantum simulation~\cite{Lloyd1996,Childs2021,Fauseweh2024}, and (iii) classical numerical simulation~\cite{Hairer2006}.
One of the most commonly used approximation methods across all disciplines is the Trotter product formula.
This is mainly for two reasons.
First, the Trotter product formula is a generic way to implement intricate quantum dynamics only by switching quickly between simpler dynamics.
Second, the accuracy of the Trotter approximation can be tuned to a desired precision by optimizing the small evolution times of each step.

To assess the simulation precision in practice, it is necessary to study the Trotter error, i.e.\ the difference between the Trotterized and target evolutions.
Due to its fundamental importance in simulations of quantum systems, countless work has been done on studying the Trotter error.
The current state-of-the-art is presented in Ref.~\cite{Childs2021} for the norm Trotter error and in Ref.~\cite{Burgarth2023} for the state-dependent Trotter error.
These results give upper bounds and explicitly quantify the worst-case scenario.
However, they often do not capture the \emph{actual} Trotter error well and tend to overestimate it~\cite{Heyl2019,Childs2019,Tran2020,Richter2021,Layden2022}\@.
This creates an urge for lower bounds specifying the best-case scenario.
Such lower bounds would provide a possible corridor for the true Trotter error, enabling a solid simulation precision analysis.
Furthermore, as lower bounds quantify the minimum possible error, they could rule out the practical usefulness of certain quantum simulation applications, or enable a complexity comparison among certain classical and quantum algorithms.
Nevertheless, little progress has been made in recent years in bounding  algorithmic complexities from below.
This is because the derivation of non-trivial lower bounds has turned out to be extremely challenging.
Here, we present explicit lower bounds on the Trotter errors, both in norm and on eigenstates.

Such lower bounds are relevant from different perspectives:
(i) Hamiltonian complexity: Combining the known upper bounds with lower bounds shows the complexity of implementing Trotterization in terms of the simulation parameters.
So far, only circuit~\cite{Knill1995,Bullock2004,Welch2014,Jia2023,Low2023} and query lower bounds~\cite{Berry2006,Childs2010,Berry2014} are known for general quantum circuits approximating a unitary $\rme^{-\rmi t H}$.
These do not answer the question of the actual simulation complexity in terms of the relevant Trotter simulation parameters.
In particular, the circuit lower bounds do not specify how large the Trotter error becomes, but only  bound below the gate complexity in terms of a desired simulation accuracy.
The query lower bounds (no fast-forwarding theorem) only imply a very loose complexity of $\Omega(t)$, which is not quantitative and does not even depend on the number of Trotter steps $n$.
Here, we show that the norm Trotter error for a finite-dimensional system scales as $\Omega(\max\{(t^2-t^3)/n,t/n-t^2/n^2\})$.
(ii) Resource estimates for simulation algorithms: Lower bounds on the Trotter errors give a strict limit on how many Trotter steps $n$ are necessary \emph{at least} to achieve a desired simulation accuracy.
This is particularly useful for resource estimates for digital \emph{and} analog quantum simulation on near-term quantum devices with limited capabilities~\cite{Childs2018,Childs2021}\@.
For these, lower bounds on the Trotter errors can determine whether the execution of a digital or analog quantum simulation algorithm is feasible.
By providing explicit bounds, our results enable the performance of such resource estimates.
(iii) Mathematical physics: The study of asymptotic Trotter convergence has a long history in mathematical physics~\cite{Kato1978,Zagrebnov2024}\@.
Here, important questions are: under which conditions does the Trotter product formula converge, in which topology, and which convergence speed does it admit?
We prove that the Trotter limit converges uniformly as $\Theta(1/n)$ for finite-dimensional systems.
(iv) Quantum speed limits: Speed limits are fundamental limitations for steering a quantum system from one state to another.
They have been used to determine the minimum time to control a quantum system~\cite{Caneva2009,Zhang2023} or to run a quantum computation~\cite{Lloyd2000,Lloyd2002}\@.
Our bounds give a fundamental precision limit for Trotterization and thus a speed limit for its convergence: we can do as much as $1/n$.
In turn, we obtain a speed limit for the convergence of Trotter-based quantum control and quantum computation schemes.

\emph{Setting.}---We focus on the case where a Hamiltonian is decomposed into two realizable components as the behavior of the Trotter error remains insufficiently characterized already in this context~\cite{Layden2022}\@.
Then, Trotterization aims to implement the time evolution $\rme^{-\rmi t (A+B)}$ under a Hamiltonian $H=A+B$ for some time $t\in\mathbb{R}$.
This is done by alternating between the dynamics under $A$ and $B$ for short time steps $t/n$, i.e.\ $\rme^{-\rmi (A+B) t}\approx(\rme^{-\rmi At/n}\rme^{-\rmi Bt/n})^n$ with $n\in\mathbb{N}$.
In the language of analog quantum simulation, one Trotter cycle $\rme^{-\rmi At/n}\rme^{-\rmi Bt/n}$ represents the evolution for one driving period.
Here, the individual evolutions $\rme^{-\rmi At/n}$ and $\rme^{-\rmi Bt/n}$ are the dynamics induced by respective control pulses.
In the case of digital quantum simulation, the individual dynamics $\rme^{-\rmi At/n}$ and $\rme^{-\rmi Bt/n}$ are realized on qubits, decomposed into elementary quantum gates.
For classical simulation algorithms, exponentiation of digital representations of $A$ and $B$ can usually be performed efficiently. See, for instance, the split-step algorithm~\cite{Hairer2006}\@.

In the following, we assume that 
$A$ and $B$ are bounded Hermitian operators acting on the Hilbert space $\mathcal{H}$.
For an eigenstate $\varphi$ of the Hermitian operator $A+B$ satisfying $(A+B)\varphi=h\varphi$, we define the state-dependent Trotter error by
\begin{equation*}
	\xi_n(t;\varphi)\equiv \left\Vert \left(\rme^{-\rmi\frac{t}{n}A}\rme^{-\rmi\frac{t}{n}B}\right)^n\varphi-\rme^{-\rmi th}\varphi \right\Vert,
	\label{eq:def_state_bound}
\end{equation*}
where $\Vert\psi\Vert=\sqrt{\langle\psi|\psi\rangle}$ is the norm of $\psi\in\mathcal{H}$.
Analogously, the norm Trotter error is defined by
\begin{equation*}
	b_n(t)\equiv \left\Vert \left(\rme^{-\rmi\frac{t}{n}A}\rme^{-\rmi\frac{t}{n}B}\right)^n-\rme^{-\rmi t(A+B)} \right\Vert_\op,
\end{equation*}
where $\Vert X\Vert_\op=\sup_{\Vert\psi\Vert=1}\Vert X\psi\Vert$ denotes the operator norm.
The operator norm $\Vert X\Vert_\op$ gives the largest singular value of $X$, and we have
	$\Vert X\Vert_\op \geq \Vert X\psi\Vert$, for all normalized 
	$\psi\in\mathcal{H}$.
To describe the asymptotic behaviors of the Trotter errors in terms of a system or simulation parameter $x$, one uses the notation $\mathcal{O}(f(x))$ for the scaling of an upper bound and $\Omega(g(x))$ for the scaling of a lower bound.
$\mathcal{O}(f(x))$ indicates that for sufficiently large $x$, the Trotter error grows no faster than a constant multiple of $f(x)$.
Analogously, $\Omega(g(x))$ means that the Trotter error grows at least as fast as a constant multiple of $g(x)$ for a sufficiently large $x$.
If the scalings of the upper and lower bounds agree, i.e.~if $f=g$ up to a constant factor, the result is tight and one writes $\Theta(f(x))$.
It is known that $\xi_n(t;\varphi)=\mathcal{O}(t^2/n)$~\cite{Burgarth2023a,Burgarth2023} and $b_n(t)=\mathcal{O}(t^2/n)$~\cite{Suzuki1985,Childs2021}\@.
In particular, we have for all $g\in\mathbb{R}$~\cite{Burgarth2023a,Burgarth2023}
\begin{equation}
\xi_n(t;\varphi)\leq \frac{t^2}{2n}\,\Bigl(
\Vert (A-h+g)^2\varphi\Vert)+\Vert (B-g)^2\varphi\Vert
\Bigr)
\label{eq:upper_state_bound}
\end{equation}
and~\cite{Suzuki1985,Childs2021}
\begin{equation}
b_n(t)\leq\frac{t^2}{2n}\Vert[A,B]\Vert_\op.
\label{eq:upper_norm_bound}
\end{equation}
Since the bound~\eqref{eq:upper_state_bound} holds for all $g\in\mathbb{R}$, we can take the infimum over $g$ to obtain the tightest bound.

\emph{Main results.}---We now complete the picture of the Trotter error analysis by presenting lower bounds on $\xi_n(t;\varphi)$ and $b_n(t)$.
\begin{main_one}
		Consider two normalized eigenstates $\varphi$ and $\psi$ of the Hamiltonian $A+B$, i.e.\ $(A+B)\varphi=h\varphi$ and $(A+B)\psi=\kappa\psi$.
		Assume that $h\neq\kappa$~\cite{Footnote_Identity} and define the spectral gap $\lambda\equiv\vert h-\kappa\vert$.
		Then, for all $g\in\mathbb{R}$ and 
		$t\geq 0$, the state-dependent Trotter error $\xi_n(t;\varphi)$ can be  bounded from below by
		\begin{align}
			\xi_n(t;\varphi)\geq{}& \left|\sin\left(\frac{\lambda t}{2}\right)\right|\left(\frac{t}{n}\bigl|\langle \psi|A\varphi\rangle\bigr|-\frac{t^{2}}{4n^{2}}\Vert (A-g)^2\varphi\Vert\right)\nonumber\\
			&{}-\frac{t^{3}}{24n^{2}}\Vert[A,[A,B]]\Vert_\op - \frac{t^{3}}{12n^{2}}\Vert[B,[B,A]]\Vert_\op.
			\label{eqn:StateBound}
		\end{align}
\end{main_one}
This bound is proved in Sec.~\hyperlink{app:state-bounds}{B} of the Supplemental Material (SM)~\cite{SM}.
Since this bound holds for all $g\in\mathbb{R}$, we can take the supremum over $g$ to obtain the tightest bound.
This involves taking the $\inf_{g}\Vert (A-g)^2\varphi\Vert$.
Furthermore, we can take the supremum over all suitable eigenstates $\psi$ to obtain the tightest result.
This bound shows that the state-dependent Trotter error $\xi_n(t;\varphi)$ scales as $\Omega(t/n-t^2/n^2)$.
Combining this result with the upper bound~(\ref{eq:upper_state_bound}) proves that the state-dependent Trotter error $\xi_n(t;\varphi)$ diminishes as $\Theta(1/n)$ on all eigenstates $\varphi$ of $A+B$.
We can always choose the Trotter number $n$ big enough to make the lower bound~\eqref{eqn:StateBound} non-trivial (larger than zero).
Such $n$ can be explicitly calculated as shown in the SM~\cite{SM}.
See Eq.~(\ref{eq:state_bound_non-trivial}) of Sec.~\hyperlink{app:state-bounds}{B}\@.

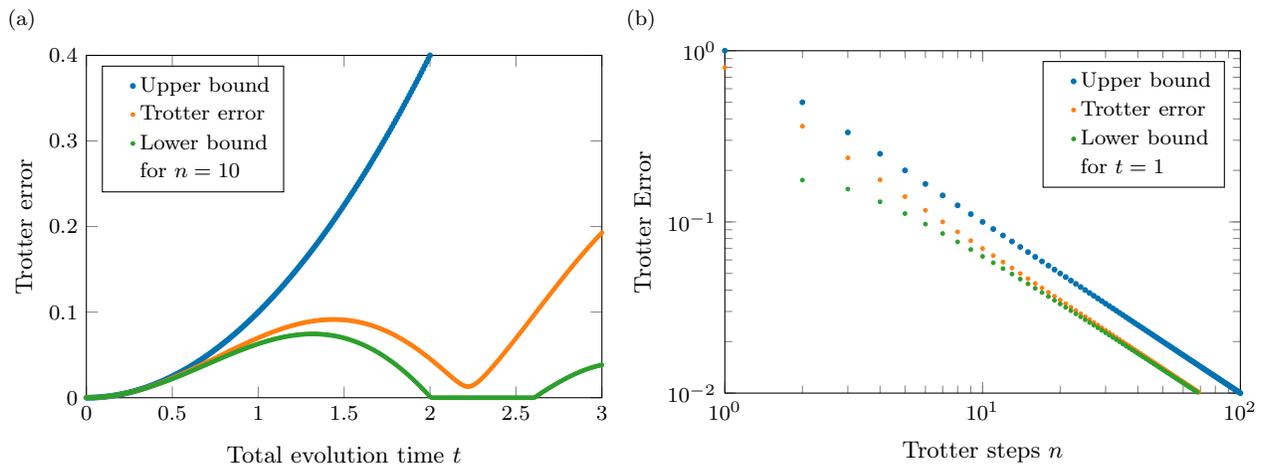
\begin{figure*}
\centering
\begin{tabular}{rr}
\multicolumn{1}{l}{\footnotesize(a)}
&
\multicolumn{1}{l}{\footnotesize(b)}
\\
\begin{tikzpicture}[mark size={0.7}, scale=1]
	\begin{axis}[
	xlabel={Total evolution time $t$},
	ylabel={Trotter error},
	tick label style={font=\footnotesize},
	x post scale=1.0,
	y post scale=0.8,
	legend cell align={left},
	legend pos=north west,
	ylabel near ticks,
	xmin=0,
	xmax=3,
	ymin=0,
	ymax=0.4,
	]
	\addplot[color=plotblue, thick, mark=*, only marks] table[x=t, y=error, col sep=comma] {One_Qubit_Time_Upper_Bound.csv};
	\addplot[color=plotorange, mark=*, only marks] table[x=t, y=error, col sep=comma] {One_Qubit_Time_Error.csv};
	\addplot[color=plotgreen, mark=*, only marks] table[x=t, y=error, col sep=comma] {One_Qubit_Time_Lower_Bound.csv};
	\addlegendimage{empty legend};
	\legend{
	\footnotesize Upper bound,
	\footnotesize Trotter error,
	\footnotesize Lower bound,
	\footnotesize for $n=10$
	};
	\end{axis}
\end{tikzpicture}
&
\begin{tikzpicture}[mark size={0.7}, scale=1]
	\begin{axis}[
	xlabel={Trotter steps $n$},
	ylabel={Trotter Error},
	tick label style={font=\footnotesize},
	x post scale=1.0,
	y post scale=0.8,
	xmode = log,
	ymode = log,
	legend cell align={left},
	legend pos=north east,
	ymin = 0,
	ylabel near ticks,
	xmin=1,
	xmax=100,
	ymin=0.01,
	ymax=1,
	]
	\addplot[color=plotblue, thick, mark=*, only marks] table[x=n, y=error, col sep=comma] {One_Qubit_Steps_Upper_Bound.csv};
	\addplot[color=plotorange, mark=*, only marks] table[x=n, y=error, col sep=comma] {One_Qubit_Steps_Error.csv};
	\addplot[color=plotgreen, mark=*, only marks] table[x=n, y=error, col sep=comma] {One_Qubit_Steps_Lower_Bound.csv};
	\addlegendimage{empty legend};
	\legend{
	\footnotesize Upper bound,
	\footnotesize Trotter error,
	\footnotesize Lower bound,
	\footnotesize for $t=1$
	};
	\end{axis}
\end{tikzpicture}
\end{tabular}
\begin{justify}
\caption{
Trotter errors for the single-qubit Pauli matrices $A=X$ and $B=Z$.
We show upper bounds (blue), the actual Trotter errors (orange), and our lower bounds (green).
For the state-dependent error $\xi_n(t;\varphi)$, we choose $\varphi$ to be the ground state of $H=A+B$, i.e.\ $H\varphi=-\sqrt{2}\varphi$.
The upper and lower bounds for both state-dependent Trotter error $\xi_n(t;\varphi)$ and norm Trotter error $b_n(t)$ are given in Eq.~\eqref{eq:1qubit_state_norm_bound} in the SM~\cite{SM}\@.
Notice that the respective bounds for the norm error and the state-dependent error coincide.
Furthermore, we find that the numerical values for the norm and state-dependent Trotter error are the same.
Therefore, the plots capture both the state-dependent and the norm error.
In all cases, our lower bound is well close to the actual Trotter error.
(a) The Trotter error $\xi_n(t;\varphi)$ or $b_n(t)$ as a function of the total evolution time $t$ for a fixed number of Trotter steps $n=10$.
(b) The Trotter error $\xi_n(t;\varphi)$ or $b_n(t)$ as a function of the number of Trotter steps $n$ for a fixed total evolution time $t=1$.
\label{fig:1qubit}
}
\end{justify}
\end{figure*}

\begin{main_two}
	The norm Trotter error $b_n(t)$ for the Hamiltonian $A+B$ can be  bounded from below by
\begin{align}
b_n(t)\geq{}&
\frac{t^2}{2n}
\left(
1
-
\frac{t}{2n}
\|A\|_\op
\right)\nonumber\\
&\hphantom{\frac{t^2}{2n}}
{}\times
\left(
\|[A,B]\|_\op
-
\frac{t}{2}
\|
[A+B,[A,B]]
\|_\op
\right)
\nonumber\\
&{}
-\frac{t^3}{24n^2}\|[A,[A,B]]\|_\op
-\frac{t^3}{12n^2}\|[B,[B,A]]\|_\op,
\label{eqn:NormBound}
\end{align}
for $0\le t\le2n/\|A\|_\op$.
\end{main_two}
A slightly tighter version of this bound is provided in the SM~\cite{SM}.
See Eq.~\eqref{eqn:complete_norm_bound} of Sec.~\hyperlink{app:norm-bounds}{C}\@.
Our bound shows that the norm Trotter error $b_n(t)$ scales as $\Omega((t^2-t^3)/n)$.
Together with the upper bound~(\ref{eq:upper_norm_bound}), this shows that  the norm Trotter error $b_n(t)$ diminishes uniformly as $\Theta(1/n)$.
Furthermore, the bound is sharp in the sense that it is saturated by two commuting Hamiltonians $[A,B]=0$.
For $[A,B]\neq0$, we can always make the lower bound~\eqref{eqn:NormBound} non-trivial by choosing $t$ small enough.
Such $t$ is computed in the SM~\cite{SM}.
See Eq.~(\ref{eq:norm_non-trivial}) of Sec.~\hyperlink{app:norm-bounds}{C}\@.
The lower bound on the norm Trotter error $b_n(t)$ can be further improved by taking the maximum with the lower bound on the state-dependent Trotter error $\xi_n(t;\varphi)$.
Recall that the lower bound on the state-dependent Trotter error $\xi_n(t;\varphi)$ is non-trivial for large $n$, whereas  the lower bound on the norm Trotter error $b_n(t)$ is non-trivial for small $t$.
Taking the maximum between these lower bounds allows for a wide corridor of simulation parameters where we can give a non-trivial lower bound on the Trotter error.
We now show two example systems, for which this is indeed the case, and compare our bounds with actual errors.

\emph{Examples.}---The first example we examine is a simple single-qubit system, where $A=X$ and $B=Z$ are Pauli matrices.
The lower bounds for this system are computed in the SM~\cite{SM}.
See Eq.~\eqref{eq:1qubit_state_norm_bound} in Sec.~\hyperlink{app:1-qubit_example}{D.1}\@.
For the state-dependent bound, we take the ground state of the target Hamiltonian $H=A+B$ as an input state $\varphi$, i.e.\ $H\varphi=-\sqrt{2}\varphi$.
The other eigenstate $\psi$ of $H=A+B$ belongs to the eigenvalue $\sqrt{2}$, i.e.\ $H\psi=\sqrt{2}\psi$, so that $\lambda=2\sqrt{2}$.
We compare our lower bounds as well as the known upper bounds~\cite{Burgarth2023,Burgarth2023a,Suzuki1985,Childs2021} with numerically estimated Trotter errors.
To this end, we show the Trotter error both as a function of the total evolution time $t$ and as a function of the number of Trotter steps $n$.
See Fig.~\ref{fig:1qubit}\@.
We find that our bounds are close to the actual error in all situations and over a wide range of parameters.
The lower bound is non-trivial for $n>\frac{at}{2\sqrt{2}}\left(1+\frac{2t}{a|{\sin\sqrt{2}\,t}|}\right)$ with $a=0.788903$.
For example, a simulation time of $t=1$ yields the condition $n\ge1$, and the lower bound is non-trivial for any $n$.

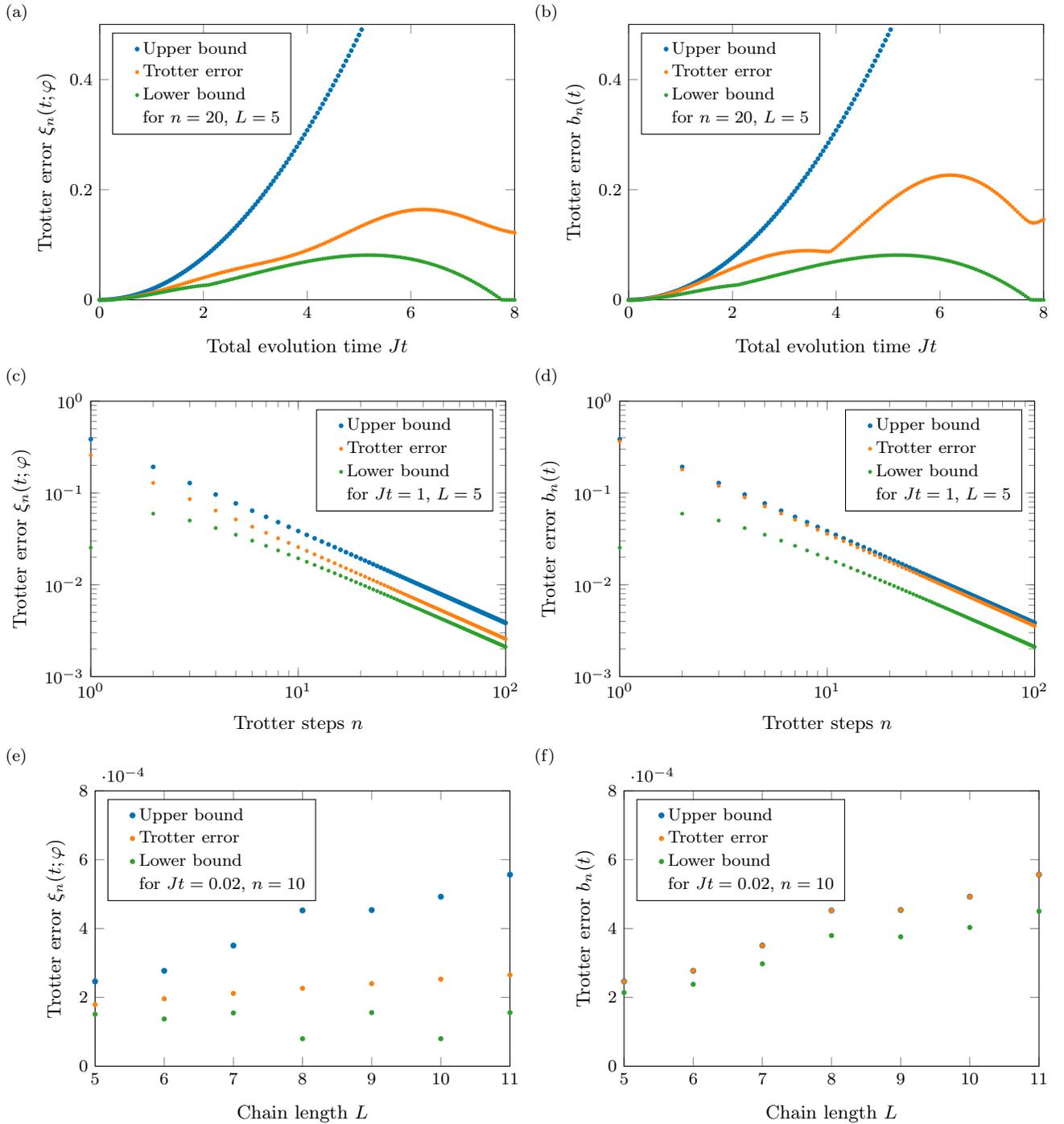
\begin{figure*}
\centering
\begin{tabular}{rr}
\multicolumn{1}{l}{\footnotesize(a)}
&
\multicolumn{1}{l}{\footnotesize(b)}
\\
	\begin{tikzpicture}[mark size={0.7}, scale=1]
			\begin{axis}[
			xlabel={Total evolution time $Jt$},
			ylabel={Trotter error $\xi_n(t;\varphi)$},
			tick label style={font=\footnotesize},
			x post scale=1.0,
			y post scale=0.8,
			legend cell align={left},
			legend pos=north west,
			ylabel near ticks,
			xmin=0,
			xmax=8,
			ymin=0,
			ymax=0.5,
			]
			\addplot[color=plotblue, thick, mark=*, only marks] table[x=t, y=error, col sep=comma] {XY_State_Time_Upper_Bound.csv};
			\addplot[color=plotorange, mark=*, only marks] table[x=t, y=error, col sep=comma] {XY_State_Time_Error.csv};
			\addplot[color=plotgreen, mark=*, only marks] table[x=t, y=error, col sep=comma] {XY_State_Time_Lower_Bound.csv};
			\addlegendimage{empty legend};
			\legend{
			\footnotesize Upper bound,
			\footnotesize Trotter error,
			\footnotesize Lower bound,
			\footnotesize {for $n=20$, $L=5$}
			};
			\end{axis}
		\end{tikzpicture}
&
	\begin{tikzpicture}[mark size={0.7}, scale=1]
			\begin{axis}[
			xlabel={Total evolution time $Jt$},
			ylabel={Trotter error $b_n(t)$},
			tick label style={font=\footnotesize},
			x post scale=1.0,
			y post scale=0.8,
			legend cell align={left},
			legend pos=north west,
			ylabel near ticks,
			yticklabel pos=left,
			xmin=0,
			xmax=8,
			ymin=0,
			ymax=0.5,
			]
			\addplot[color=plotblue, thick, mark=*, only marks] table[x=t, y=error, col sep=comma] {XY_Norm_Time_Upper_Bound.csv};
			\addplot[color=plotorange, mark=*, only marks] table[x=t, y=error, col sep=comma] {XY_Norm_Time_Error.csv};
			\addplot[color=plotgreen, mark=*, only marks] table[x=t, y=error, col sep=comma] {XY_Norm_Time_Lower_Bound.csv};
			\addlegendimage{empty legend};
			\legend{
			\footnotesize Upper bound,
			\footnotesize Trotter error,
			\footnotesize Lower bound,
			\footnotesize {for $n=20$, $L=5$}
			};
			\end{axis}
		\end{tikzpicture}
\\
\multicolumn{1}{l}{\footnotesize(c)}
&
\multicolumn{1}{l}{\footnotesize(d)}
\\
	\begin{tikzpicture}[mark size={0.7}, scale=1]
			\begin{axis}[
			xlabel={Trotter steps $n$},
			ylabel={Trotter error $\xi_n(t;\varphi)$},
			tick label style={font=\footnotesize},
			x post scale=1.0,
			y post scale=0.8,
			xmode = log,
			ymode = log,
			legend cell align={left},
			legend pos=north east,
			ymin = 0,
			ylabel near ticks,
			xmin=1,
			xmax=100,
			ymax=1,
			ymin=0.001,
			]
			\addplot[color=plotblue, thick, mark=*, only marks] table[x=n, y=error, col sep=comma] {XY_State_Steps_Upper_Bound.csv};
			\addplot[color=plotorange, mark=*, only marks] table[x=n, y=error, col sep=comma] {XY_State_Steps_Error.csv};
			\addplot[color=plotgreen, mark=*, only marks] table[x=n, y=error, col sep=comma] {XY_State_Steps_Lower_Bound.csv};
			\addlegendimage{empty legend};
			\legend{
			\footnotesize Upper bound,
			\footnotesize Trotter error,
			\footnotesize Lower bound,
			\footnotesize {for $Jt=1$, $L=5$}
			};
			\end{axis}
		\end{tikzpicture}
&
	\begin{tikzpicture}[mark size={0.7}, scale=1]
			\begin{axis}[
			xlabel={Trotter steps $n$},
			ylabel={Trotter error $b_n(t)$},
			tick label style={font=\footnotesize},
			x post scale=1.0,
			y post scale=0.8,
			xmode = log,
			ymode = log,
			legend cell align={left},
			legend pos=north east,
			ylabel near ticks,
			yticklabel pos=left,
			xmin=1,
			xmax=100,
			ymax=1,
			ymin=0.001,
			]
			\addplot[color=plotblue, thick, mark=*, only marks] table[x=n, y=error, col sep=comma] {XY_Norm_Steps_Upper_Bound.csv};
			\addplot[color=plotorange, mark=*, only marks] table[x=n, y=error, col sep=comma] {XY_Norm_Steps_Error.csv};
			\addplot[color=plotgreen, mark=*, only marks] table[x=n, y=error, col sep=comma] {XY_Norm_Steps_Lower_Bound.csv};
			\addlegendimage{empty legend};
			\legend{
			\footnotesize Upper bound,
			\footnotesize Trotter error,
			\footnotesize Lower bound,
			\footnotesize {for $Jt=1$, $L=5$}
			};
			\end{axis}
		\end{tikzpicture}
\\
\multicolumn{1}{l}{\footnotesize(e)}
&
\multicolumn{1}{l}{\footnotesize(f)}
\\
	\begin{tikzpicture}[mark size={1}, scale=1]
			\begin{axis}[
			xlabel={Chain length $L$},
			ylabel={Trotter error $\xi_n(t;\varphi)$},
			tick label style={font=\footnotesize},
			x post scale=1.0,
			y post scale=0.8,
			legend cell align={left},
			legend pos=north west,
			ylabel near ticks,
			xmin=5,
			xmax=11,
			ymin=0,
			ymax=0.0008,
			]
			\addplot[color=plotblue, thick, mark=*, only marks] table[x=L, y=error, col sep=comma] {XY_State_Length_Upper_Bound.csv};
			\addplot[color=plotorange, mark=*, only marks] table[x=L, y=error, col sep=comma] {XY_State_Length_Error.csv};
			\addplot[color=plotgreen, mark=*, only marks] table[x=L, y=error, col sep=comma] {XY_State_Length_Lower_Bound.csv};
			\addlegendimage{empty legend};
			\legend{
			\footnotesize Upper bound,
			\footnotesize Trotter error,
			\footnotesize Lower bound,
			\footnotesize {for $Jt=0.02$, $n=10$}
			};
			\end{axis}
		\end{tikzpicture}
&
	\begin{tikzpicture}[mark size={1}, scale=1]
			\begin{axis}[
			xlabel={Chain length $L$},
			ylabel={Trotter error $b_n(t)$},
			tick label style={font=\footnotesize},
			x post scale=1.0,
			y post scale=0.8,
			legend cell align={left},
			legend pos=north west,
			ylabel near ticks,
			yticklabel pos=left,
			xmin=5,
			xmax=11,
			ymin=0,
			ymax=0.0008,
			]
			\addplot[color=plotblue, thick, mark=*, only marks] table[x=L, y=error, col sep=comma] {XY_Norm_Length_Upper_Bound.csv};
			\addplot[color=plotorange, mark=*, only marks] table[x=L, y=error, col sep=comma] {XY_Norm_Length_Error.csv};
			\addplot[color=plotgreen, mark=*, only marks] table[x=L, y=error, col sep=comma] {XY_Norm_Length_Lower_Bound.csv};
			\addlegendimage{empty legend};
			\legend{
			\footnotesize Upper bound,
			\footnotesize Trotter error,
			\footnotesize Lower bound,
			\footnotesize {for $Jt=0.02$, $n=10$}
			};
			\end{axis}
		\end{tikzpicture}
\end{tabular}
\begin{justify}
\caption{Trotter errors for the $XX$ model $H=A+B$ in Eq.~\eqref{eq:XX-model} with $A = \frac{J}{4} \sum_{j=0}^{-1}  X_j X_{j+1}$ and $B = \frac{J}{4} \sum_{j=0}^{L-1}  Y_j Y_{j+1}$.
We show upper bounds (blue), the actual Trotter errors (orange), and our lower bounds (green) for the two different scenarios.
Slightly looser, but explicit versions of the upper and lower bounds for both state-dependent Trotter error $\xi_n(t;\varphi)$ and norm Trotter error $b_n(t)$ are given in Eqs.~(\ref{eq:chain_state_bound}) and~(\ref{eq:chain_norm_bound}) in the SM~\cite{SM}, respectively. 
In the state-dependent case, we consider the eigenstate $\varphi=\ket{\downarrow\cdots\downarrow}$.
In all cases, our lower bounds capture the true error very accurately.
(a) The state-dependent Trotter error $\xi_n(t;\varphi)$ as a function of the total evolution time $Jt$ for a fixed number of Trotter steps $n=20$ and a chain length $L=5$.
(b) The norm Trotter error $b_n(t)$ as a function of the total evolution time $Jt$ for a fixed number of Trotter steps $n=20$ and a chain length $L=5$.
(c) The state-dependent Trotter error $\xi_n(t;\varphi)$ as a function of the number of Trotter steps $n$ for a fixed total evolution time $Jt=1$ and a chain length $L=5$.
(d) The norm Trotter error $b_n(t)$ as a function of the number of Trotter steps $n$ for a fixed total evolution time $Jt=1$ and a chain length $L=5$.
(e) The state-dependent Trotter error $\xi_n(t;\varphi)$ as a function of the chain length $L$ for a fixed total evolution time $Jt=0.02$ and a number of Trotter steps $n=10$.
(f) The norm Trotter error $b_n(t)$ as a function of the chain length $L$ for a fixed total evolution time $Jt=0.02$ and a number of Trotter steps $n=10$.
\label{fig:XX}
}
\end{justify}
\end{figure*}

As a second example, we consider the isotropic one-dimensional $XX$ model of chain length $L$,
\begin{equation}
    H = \frac{J}{4} \sum_{j=0}^{L-1}( X_j X_{j+1} +  Y_j Y_{j+1}).
    \label{eq:XX-model}
\end{equation}
Here, $J\in\mathbb{R}$ is the coupling constant and we impose the periodic boundary condition.
We Trotterize between the non-commuting elements $A = \frac{J}{4} \sum_{j=0}^{L-1}  X_j X_{j+1}$ and $B = \frac{J}{4} \sum_{j=0}^{L-1}  Y_j Y_{j+1}$.
This is conceptually analogous to the Ising model simulations performed in Ref.~\cite{Heyl2019}\@.
The lower and upper bounds on the Trotter errors $\xi_n(t;\varphi)$ and $b_n(t)$ for this model are computed in the SM~\cite{SM}. See Eqs.~(\ref{eq:chain_state_bound}) and~(\ref{eq:chain_norm_bound}) of Sec.~\hyperlink{app:XX_example}{D.2}, respectively.
For the state-dependent case, we study the Trotter error $\xi_n(t;\varphi)$ for the state $\varphi=\ket{\downarrow\cdots\downarrow}$ with no spin excitations, which has energy $h=0$.
To compute its lower bound, we pick an excited state $\psi$ with two spin excitations and optimize its parameter to maximize the lower bound.
The optimal parameter might change depending on the total evolution time $t$ due to the factor $|{\sin(\lambda t/2)}|$ in the lower bound on $\xi_n(t;\varphi)$ in Eq.~(\ref{eqn:StateBound}).
We compare our lower bounds and the known upper bounds~\cite{Burgarth2023,Burgarth2023a,Suzuki1985,Childs2021} with numerical simulations.
We show the Trotter errors $\xi_n(t;\varphi)$ and $b_n(t)$ as functions of the total evolution time $t$, the number of Trotter steps $n$, and the chain length $L$.
See Fig.~\ref{fig:XX}~\footnote{For Fig.~\ref{fig:XX} and Fig.~\ref{fig:QPE}, we used the tightest bounds we could obtain without further estimating the norms appearing in the bounds. Notice that Eqs.~(\ref{eq:chain_state_bound}) and~(\ref{eq:chain_norm_bound}) in the SM~\cite{SM} are slightly looser as their derivation involves norm equivalences and triangle inequalities.}.
We find that our bounds capture the true error very accurately and are non-trivial in a large parameter regime.
In particular, our lower bound on the norm Trotter error $b_n(t)$ captures the fact that the Trotter error increases with increasing system size $L$ for small enough $Jt$.
This suggests that our bounds are of high practical relevance for the simulation of quantum many-body systems.
Indeed, they also give useful error estimates in the context of quantum phase estimation for Hamiltonian simulation.
Here, we have to choose $t<t_\mathrm{QPE}\equiv 1/\Vert H\Vert_\op$ to prevent the eigenvalues of $H$ from encompassing a full unit circle when evolving under the Trotter unitary~\cite{Bauer2014}\@.
We compute $t_\mathrm{QPE}$ and the maximum time $t_{b_n}$, for which the lower bound on $b_n(t)$ is non-trivial, for different chain lengths $L$.
We find that $t_{b_n}\gg t_\mathrm{QPE}$ already for small $n$, so that the lower bound on $b_n(t)$ becomes non-trivial in all physically relevant settings.
See Fig.~\ref{fig:QPE} for the numerical results.
This showcases how our bounds can be immediately applied to the error analysis in quantum phase estimation.
\begin{figure}
\centering
	\begin{tikzpicture}[mark size={0.7}, scale=1]
			\begin{axis}[
			xlabel={Trotter steps $n$},
			ylabel={$t_{b_n}$ and $t_\mathrm{QPE}$\ \ ($J^{-1}$)},
			tick label style={font=\footnotesize},
			x post scale=1.0,
			y post scale=0.8,
			legend columns=2,
			legend cell align={left},
			legend pos=north west,
			legend style={font=\footnotesize},
			ylabel near ticks,
			xmin=1,
			xmax=8,
			ymin=0,
			ymax=5,
			]
			\addplot[color=plotblue, thick, mark=*, only marks] table[x=n, y=L1, col sep=comma] {XY_QPE.csv};
			\addlegendentry{$L=5$};
			\addplot[color=plotorange, mark=*, only marks] table[x=n, y=L2, col sep=comma] {XY_QPE.csv};
			\addlegendentry{$L=6$};
			\addplot[color=plotgreen, mark=*, only marks] table[x=n, y=L3, col sep=comma] {XY_QPE.csv};
			\addlegendentry{$L=7$};
			\addplot[color=plotred, mark=*, only marks] table[x=n, y=L4, col sep=comma] {XY_QPE.csv};
			\addlegendentry{$L=8$};
			\addplot[color=plotpurple, mark=*, only marks] table[x=n, y=L5, col sep=comma] {XY_QPE.csv};
			\addlegendentry{$L=9$};
			\addplot[color=plotbrown, mark=*, only marks] table[x=n, y=L6, col sep=comma] {XY_QPE.csv};
			\addlegendentry{$L=10$};
			\addplot[color=plotpink, mark=*, only marks] table[x=n, y=L7, col sep=comma] {XY_QPE.csv};
			\addlegendentry{$L=11$};
			\addplot[color=plotblue, domain=1:20] {0.154508};
			\addplot[color=plotorange, domain=1:20] {0.125};
			\addplot[color=plotgreen, domain=1:20] {0.11126};
			\addplot[color=plotred, domain=1:20] {0.0956709};
			\addplot[color=plotpurple, domain=1:20] {0.0868241};
			\addplot[color=plotbrown, domain=1:20] {0.0772542};
			\addplot[color=plotpink, domain=1:20] {0.0711574};
			\end{axis}
		\end{tikzpicture}
	\begin{justify}
	\caption{The minimum time $t_{b_n}$ for which the lower bound on the norm Trotter error $b_n(t)$ becomes non-trivial (dots), and the maximum evolution time $t_\mathrm{QPE}$ for quantum phase estimation before eigenvalue circling (lines), for the $XX$ model $H=A+B$ in Eq.~\eqref{eq:XX-model} with $A=\frac{J}{4}\sum_{j=0}^{L-1} X_j X_{j+1}$ and $B=\frac{J}{4}\sum_{j=0}^{L-1} Y_j Y_{j+1}$, for different chain lengths $L$. Whenever $t_{b_n}>t_\mathrm{QPE}$, our lower bound is practically useful. We find that this is always the case. \label{fig:QPE}}
	\end{justify}
\end{figure}
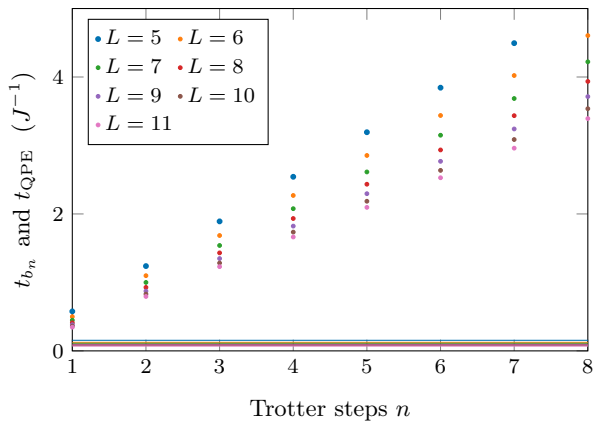

\emph{Conclusion.}---To summarize, we provide lower bounds for the Trotter error
 for bounded Hamiltonians, both in the operator norm and on eigenstates of the target dynamics.
The state-dependent error scales as $\Omega\big(t/n-t^2/n^2\big)$ and the norm error admits a $\Omega\big(\max\{(t^2-t^3)/n,t/n-t^2/n^2\}\big)$ scaling.
Overall, this proves that the finite-dimensional Trotter product formula converges as $\Theta(1/n)$.
Our bounds are tested numerically showing that they genuinely estimate the the true Trotter error.
To the best of our knowledge, this is the first time that the Trotter error is non-trivially bounded from below in terms of all simulation parameters.
This enables a better understanding of practical simulation errors, allowing a more precise fine-tuning of parameters in both digital and analog simulations.
This is particularly important for quantum simulation tasks, which are believed to be one of the most promising potential applications of quantum information science.

\bigskip
\acknowledgments
The authors thank Christian Arenz and Dominic Berry for their helpful feedback on the manuscript.
Furthermore, AH thanks Dominic Berry and Lauritz van Luijk for interesting discussions.

AH was partially supported by the Sydney Quantum Academy.
DB acknowledges funding from the Australian Research Council (project numbers FT190100106, DP210101367, CE170100009) and the Munich Quantum Valley project K8\@.
PF acknowledges supports from the PNRR MUR project CN00000013 - National Centre for HPC, Big Data and Quantum Computing, from Istituto Nazionale di Fisica Nucleare (INFN) through the project ``QUANTUM'', from the Italian National Group of Mathematical Physics (GNFM-INdAM), and from the Italian funding
within the ``Budget MUR - Dipartimenti di Eccellenza 2023--2027'' - Quantum Sensing
and Modelling for One-Health (QuaSiModO).
KY acknowledges supports by the Top Global University Project from the Ministry of Education, Culture, Sports, Science and Technology (MEXT), Japan, and by JSPS KAKENHI Grant Numbers JP18K03470, JP18KK0073, and JP24K06904, from the Japan Society for the Promotion of Science (JSPS).

\bibliographystyle{prsty-title-hyperref}
\bibliography{Trotter_Lower_Bounds}

\begin{thebibliography}{10}
\bibitem{Footnote_Identity}
To ensure that a state $\psi$ with $\kappa\neq h$ exists, we may assume that
  $A+B\not\propto I$, where $I$ is the identity operator.
\end{thebibliography}

\newpage
\onecolumngrid
\appendix

\begin{center}
	{\Large\textbf{Supplemental Material}}
\end{center}

\hypertarget{app:preliminaries}{}
\section{A. Notation, Proof Idea, and Basic Lemma}
In this section, we introduce the notation used throughout the proofs of the main results.
Throughout the entire Supplemental Material, we will assume that $A$ and $B$ are bounded Hamiltonians (Hermitian operators) acting on a
Hilbert space $\mathcal{H}$.
They generate unitary one-parameter groups $\rme^{-\rmi tA}$ and $\rme^{-\rmi tB}$, respectively, where the parameter $t\in\mathbb{R}$ is the time.
We introduce the operator norm $\Vert X\Vert_\op=\sup_{\Vert \psi\Vert=1} \Vert X\psi\Vert$, where $\Vert\psi\Vert=\sqrt{\langle \psi,\psi\rangle}$ is the
norm of vectors in $\mathcal{H}$.
For matrices, the operator norm computes to the largest singular value.
It is a unitarily invariant matrix norm.
We remark that our proof method works for any (submultiplicative)
norm as long as it is unitarily invariant.

In the following two sections, we will provide proofs for the lower bounds on the Trotter errors, one in norm and the other on an eigenstate.
The two proofs follow the same idea, which we would like to present here informally.
Let
\begin{equation}
T_{n}(t) \equiv\left(\mathrm{e}^{-\mathrm{i}\frac{t}{n}A}\mathrm{e}^{-\mathrm{i}\frac{t}{n}B}\right)^{n}
\label{eq:def_first_order}
\end{equation}
be the first-order Trotter product formula, let
\begin{equation}
	T_{n}^{(2)}(t) \equiv\left(\mathrm{e}^{-\mathrm{i}\frac{t}{2n}A}\mathrm{e}^{-\mathrm{i}\frac{t}{n}B}\mathrm{e}^{-\mathrm{i}\frac{t}{2n}A}\right)^{n}
	\label{eq:def_second_order}
\end{equation}
be the second-order Trotter product formula, and let
\begin{equation*}
	T(t) \equiv\mathrm{e}^{-\mathrm{i}t(A+B)}
\end{equation*}
be the target evolution that we aim to approximate through Trotterization.
First, we observe that the first-order Trotter product~\eqref{eq:def_first_order} and the second-order Trotter product~\eqref{eq:def_second_order} only differ by boundary terms, i.e.,
\begin{equation*}
T_{n}^{(2)}(t)=\mathrm{e}^{\mathrm{i}\frac{t}{2n}A}T_{n}(t)\mathrm{e}^{-\mathrm{i}\frac{t}{2n}A}.
\end{equation*}
This has also been noticed in Ref.~\cite{Layden2022}, where the author used this fact to prove tighter upper bounds for the first-order Trotterization.
From Ref.~\cite{Childs2021}, we know that
\[
T_{n}^{(2)}(t)-T(t)=\mathcal{O}\!\left(\frac{t^{3}}{n^{2}}\right),
\]
which is equivalent to
\[
\mathrm{e}^{\mathrm{i}\frac{t}{2n}A}T_{n}(t)\mathrm{e}^{-\mathrm{i}\frac{t}{2n}A}-T(t)=\mathcal{O}\!\left(\frac{t^{3}}{n^{2}}\right).
\]
In turn,
\begin{equation}
T_{n}(t)-\mathrm{e}^{-\mathrm{i}\frac{t}{2n}A}T(t)\mathrm{e}^{\mathrm{i}\frac{t}{2n}A}=\mathcal{O}\!\left(\frac{t^{3}}{n^{2}}\right)
\label{eq:proof_idea_boundary}.
\end{equation}
Since Eq.~\eqref{eq:proof_idea_boundary} is dominated by any term that scales as $\mathcal{O}(1/n)$, it is useful for a reverse triangle inequality.
In particular, we can insert a zero in the Trotter error as
\begin{equation*}
	T_n(t)-T(t) = T_n(t) - \mathrm{e}^{-\mathrm{i}\frac{t}{2n}A}T(t)\mathrm{e}^{\mathrm{i}\frac{t}{2n}A} + \mathrm{e}^{-\mathrm{i}\frac{t}{2n}A}T(t)\mathrm{e}^{\mathrm{i}\frac{t}{2n}A} - T(t).
\end{equation*}
By the reverse triangle inequality, the Trotter error can be  bounded below by
\begin{equation}
	b_n(t)\equiv\Vert T_n(t)-T(t)\Vert_\op \geq
	\left\Vert \mathrm{e}^{-\mathrm{i}\frac{t}{2n}A}T(t)\mathrm{e}^{\mathrm{i}\frac{t}{2n}A} - T(t) \right\Vert_\op
	-
	\left\Vert T_n(t) - \mathrm{e}^{-\mathrm{i}\frac{t}{2n}A}T(t)\mathrm{e}^{\mathrm{i}\frac{t}{2n}A} \right\Vert_\op
	\label{eq:proof_norm_bounds_start}
\end{equation}
in norm, and by
\begin{equation}
	\xi_n(t;\varphi)\equiv\Vert [T_n(t)-T(t)]\varphi\Vert \geq
	\left\Vert \left(\mathrm{e}^{-\mathrm{i}\frac{t}{2n}A}T(t)\mathrm{e}^{\mathrm{i}\frac{t}{2n}A} - T(t)\right)\varphi \right\Vert
	-
	\left\Vert \left(T_n(t) - \mathrm{e}^{-\mathrm{i}\frac{t}{2n}A}T(t)\mathrm{e}^{\mathrm{i}\frac{t}{2n}A}\right)\varphi \right\Vert\label{eq:proof_state_bounds_start}
\end{equation}
for vectors $\varphi\in\mathcal{H}$.
The proof of the bounds then reduces to the technical task of bounding these two terms.

To do so, we will use the notion of the adjoint representation:
Given a bounded Hamiltonian $H$ acting on the Hilbert space $\mathcal{H}$, its adjoint representation $\ad_H$ is given by the commutator $\ad_H=[H,{}\bullet{}]$.
In the same way as $H$ generates a one-parameter group $U(t)=\rme^{-\rmi tH}$, the adjoint representation $\ad_H$ does via $\Ad_{U(t)}=\rme^{-\rmi t\ad_H}$.
It is well known that this can be rewritten as $\Ad_{U(t)}=\rme^{-\rmi tH}\bullet\rme^{\rmi tH}$.
For a (bounded) Hamiltonian $H$, the adjoint representation $\ad_H$ is the (bounded) generator of a group of isometries on the Banach space $\mathcal{B}(\mathcal{H})$ of bounded linear operators.
Therefore, the same algebraic relations as for $H$ and $U(t)$ on $\mathcal{H}$ also hold for $\ad_H$ and $\Ad_{U(t)}$ on $\mathcal{B}(\mathcal{H})$.
See Ref.~\cite{Lonigro2024} for details.
This fact will turn out to be very useful for proving the bound on the Trotter error in norm.
For example, we will use Ref.~\cite[Lemma~14]{Hahn2022}, which also holds for the adjoint representation.
Let us recapitulate Ref.~\cite[Lemma~14]{Hahn2022} here and extend its statement to the adjoint representation.

\begin{lem}\label{lem:basic_lemma}
	Let $H=H^\dagger$ and $t\geq 0$. Then,
	\begin{align*}
		\rme^{-\rmi tH} - I &= -\rmi\int_0^t\rme^{-\rmi sH} H\,\rmd s,\\
		\rme^{-\rmi tH} - I + \rmi tH &= -\int_0^t (t-s) \rme^{-\rmi sH} H^2\,\rmd s,
	\end{align*}
	and analogously for the adjoint representation
	\begin{align*}
		\rme^{-\rmi t\ad_H} - \id &= -\rmi\int_0^t\rme^{-\rmi s\ad_H}\ad_H\rmd s,\\
		\rme^{-\rmi t\ad_H} - \id + \rmi t\ad_H &= -\int_0^t (t-s) \rme^{-\rmi s\ad_H}\ad_H^2\rmd s,
	\end{align*}
	where $I$ and $\id$ are the identity operator and the identity map, respectively.
	Therefore, by taking norms we obtain
	\begin{gather*}
		\Vert\rme^{-\rmi tH} - I \Vert_\op \leq t \Vert H\Vert_\op,\\
		\Vert\rme^{-\rmi tH} - I + \rmi tH\Vert_\op \leq \frac{1}{2}t^2 \Vert H^2\Vert_\op,
	\end{gather*}
	and analogously for all $X\in\mathcal{B}(\mathcal{H})$
	\begin{gather*}
		\Vert(\rme^{-\rmi t\ad_H} - \id)(X)\Vert_\op \leq t\Vert [H,X]\Vert_\op,\\
		\Vert(\rme^{-\rmi t\ad_H} - \id + \rmi t\ad_H)(X)\Vert_\op \leq \frac{1}{2}t^2\Vert [H,[H,X]]\Vert_\op.
	\end{gather*}
\end{lem}
\begin{proof}
	This is an application of the Taylor polynomial of the exponential function with remainder in integral form. See e.g.~Ref.~\cite[Lemma~14]{Hahn2022}\@.
\end{proof}

\hypertarget{app:state-bounds}{}
\section{B. Lower Bound on the State-Dependent Trotter Error $\xi_n(t;\varphi)$}
In this section, we prove the lower bound on the state-dependent Trotter error $\xi_n(t;\varphi)$ presented inEq.~(\ref{eqn:StateBound}) of the main text.

\begin{proof}[Proof of the Lower Bound on the State-Dependent Trotter Error $\xi_n(t;\varphi)$]
By assumption, we have $(A+B)\varphi=h\varphi$.
The first step is to shift the energy of $A+B$ so that $(A+B)\varphi=0$.
Bounds for arbitrary states with $(A+B)\varphi=h\varphi$ can then be retrieved by shifting back the energy at the end.
This idea has also been employed in Refs.~\cite{Burgarth2023a, Burgarth2023} to prove upper bounds for the Trotter error on eigenstates.

After the energy shift, we start with Eq.~\eqref{eq:proof_state_bounds_start}.
Let us look at the first term.
Using $T(t)\varphi=\varphi$, we can write
\begin{align*}
\left\Vert \left(\mathrm{e}^{-\mathrm{i}\frac{t}{2n}A}T(t)\mathrm{e}^{\mathrm{i}\frac{t}{2n}A}-T(t)\right)\varphi\right\Vert  
& =\left\Vert \left(T(t)\mathrm{e}^{\mathrm{i}\frac{t}{2n}A}-\mathrm{e}^{\mathrm{i}\frac{t}{2n}A}\right)\varphi\right\Vert \\
& =\left\Vert[T(t)-I]\left(\mathrm{e}^{\mathrm{i}\frac{t}{2n}A}-I\right)\varphi\right\Vert \\
 & \geq\left|\left\langle \psi \Big| [T(t)-I]\left(\mathrm{e}^{\mathrm{i}\frac{t}{2n}A}-I\right)\varphi\right\rangle \right|,
 \end{align*}
for all normalized vector $\psi\neq\varphi$.
By assuming that $(A+B)\psi=\kappa\psi$, we get
 \begin{align*}
\left\Vert \left(\mathrm{e}^{-\mathrm{i}\frac{t}{2n}A}T(t)\mathrm{e}^{\mathrm{i}\frac{t}{2n}A}-T(t)\right)\varphi\right\Vert  
&\ge \left|\left\langle(\mathrm{e}^{\mathrm{i}\kappa t}-1)\psi \Big| \left(\mathrm{e}^{\mathrm{i}\frac{t}{2n}A}-I\right)\varphi\right\rangle \right|\\
& =2\left|\sin\left(\frac{\kappa t}{2}\right)\right|\left|\left\langle \psi \Big| \left(\mathrm{e}^{\mathrm{i}\frac{t}{2n}A}-I\right)\varphi\right\rangle \right|.
 \end{align*}
We further proceed as
 \begin{align}
\left\Vert \left(\mathrm{e}^{-\mathrm{i}\frac{t}{2n}A}T(t)\mathrm{e}^{\mathrm{i}\frac{t}{2n}A}-T(t)\right)\varphi\right\Vert  
&\ge
2\left|\sin\left(\frac{\kappa t}{2}\right)\right|
\left|\left\langle \psi \bigg|\left(\mathrm{e}^{\mathrm{i}\frac{t}{2n}A}-I-\frac{\mathrm{i}t}{2n}A\right)\varphi\right\rangle +\frac{\mathrm{i}t}{2n}\langle \psi | A\varphi\rangle \right|\nonumber\\
 & \geq2\left|\sin\left(\frac{\kappa t}{2}\right)\right|\left(
 \frac{t}{2n}|\langle \psi | A\varphi\rangle|
 -\left|\left\langle \psi \bigg|\left(\mathrm{e}^{\mathrm{i}\frac{t}{2n}A}-I-\frac{\mathrm{i}t}{2n}A\right)\varphi\right\rangle \right|
 \right)\label{eq:state_bound_step_alternative}\\
 & \geq2\left|\sin\left(\frac{\kappa t}{2}\right)\right|\left(
 \frac{t}{2n}|\langle\psi | A\varphi\rangle|
 -\left\Vert \left(\mathrm{e}^{\mathrm{i}\frac{t}{2n}A}-I-\frac{\mathrm{i}t}{2n}A\right)\varphi\right\Vert
 \right)\nonumber\\
 & \geq2\left|\sin\left(\frac{\kappa t}{2}\right)\right|\left(
 \frac{t}{2n}|\langle\psi | A\varphi\rangle|
 -\frac{t^{2}}{8n^{2}}\Vert A^2\varphi\Vert
 \right),\label{eq:state_bound_step_alternative_2}
\end{align}
where the third step follows from the Cauchy-Schwarz inequality and the
last step is a consequence of Lemma~\ref{lem:basic_lemma}\@.
We remark that from Eq.~\eqref{eq:state_bound_step_alternative}, we could have alternatively bounded
\begin{equation*}
	\left\Vert \left(\mathrm{e}^{-\mathrm{i}\frac{t}{2n}A}T(t)\mathrm{e}^{\mathrm{i}\frac{t}{2n}A}-T(t)\right)\varphi\right\Vert \ge 2\left|\sin\left(\frac{\kappa t}{2}\right)\right|\left(
 \frac{t}{2n}|\langle\psi | A\varphi\rangle|
 -\frac{t^{2}}{8n^{2}}\Vert A\varphi\Vert \Vert A\psi\Vert
 \right).
\end{equation*}
However, in the examples considered we found that the first bound~\eqref{eq:state_bound_step_alternative_2} is tighter, which is why we use it here.

For the second term of Eq.~\eqref{eq:proof_state_bounds_start}, we can bound it as
\begin{align*}
\left\Vert \left(T_{n}(t)-\mathrm{e}^{-\mathrm{i}\frac{t}{2n}A}T(t)\mathrm{e}^{\mathrm{i}\frac{t}{2n}A}\right)\varphi\right\Vert  
& \leq \left\Vert T_{n}(t)-\mathrm{e}^{-\mathrm{i}\frac{t}{2n}A}T(t)\mathrm{e}^{\mathrm{i}\frac{t}{2n}A}\right\Vert _\op\\
& =\left\Vert \mathrm{e}^{\mathrm{i}\frac{t}{2n}A}T_{n}(t)\mathrm{e}^{-\mathrm{i}\frac{t}{2n}A}-T(t)\right\Vert _\op\\
& =\Vert T_{n}^{(2)}(t)-T(t)\Vert _\op.
\vphantom{\left\Vert \mathrm{e}^{\mathrm{i}\frac{t}{2n}A}\right\Vert}
\end{align*}
An error bound for this distance can be found in Ref.~\cite[Eq.~(L5)]{Childs2021},
\begin{equation}
\Vert T_{n}^{(2)}(t)-T(t)\Vert_\op\leq\frac{t^{3}}{24n^{2}}\Vert[A,[A,B]]\Vert_\op+\frac{t^{3}}{12n^{2}}\Vert[B,[B,A]]\Vert_\op .
\label{eqn:Childs2021}
\end{equation}
In total, we get
\begin{equation*}
\xi_n(t;\varphi)
=\Vert[T_{n}(t)-T(t)]\varphi\Vert
  \geq\left|\sin\left(\frac{\kappa t}{2}\right)\right|\left(\frac{t}{n}|\langle \psi | A\varphi\rangle|
  -\frac{t^{2}}{4n^{2}}\Vert A^{2}\varphi\Vert \right)
  -\frac{t^{3}}{24n^{2}}\Vert[A,[A,B]]\Vert_\op
  -\frac{t^{3}}{12n^{2}}\Vert[B,[B,A]]\Vert_\op.
\end{equation*}
The last step is to shift back the energy and to notice that we can always add any multiple $g\in\mathbb{R}$ of the identity to $A$ and $B$ without changing the Trotter error.
This leads to the replacements $\kappa\rightarrow\vert h - \kappa\vert$ as well as $A\rightarrow A-h+g$ and $B\rightarrow B+g$.
Since the latter two replacements only affect the term $\Vert A^2\varphi\Vert$ in the bound and $g\in\mathbb{R}$ is arbitrary, Eq.~(\ref{eqn:StateBound}) of the main text follows.
\end{proof}

To explicitly see the dependence of this bound on $t$,
we further bound
\[
\left|\sin\left(\frac{\kappa t}{2}\right)\right|\geq1-\left|1-\frac{\kappa t}{\pi}+2\left\lfloor \frac{\kappa t}{2\pi}\right\rfloor \right|,
\]
so that we have
\[
\xi_n(t;\varphi)
=\Omega\!\left(\frac{t}{n}-\frac{t^2}{n^2}\right).
\]
Together with Ref.~\cite[Main~Result~1]{Burgarth2023a}, this proves that the Trotter product converges strongly as $\Theta(1/n)$.
Since the norm is nonnegative, we have
\begin{equation*}
			\xi_n(t;\varphi)\geq \max\!\left\{0,\left|\sin\left(\frac{\kappa t}{2}\right)\right|\left(\frac{t}{n}|\langle \psi | A\varphi\rangle|
  -\frac{t^{2}}{4n^{2}}\Vert A^{2}\varphi\Vert \right) 
  -\frac{t^{3}}{24n^{2}}\Vert[A,[A,B]]\Vert_\op
  -\frac{t^{3}}{12n^{2}}\Vert[B,[B,A]]\Vert_\op\right\}.
\end{equation*}
If $[A,B]=0$, we have $\xi_n(t;\varphi)=0$, while the bound reduces to
\begin{equation*}
\xi_n(t;\varphi)
\geq\max\!\left\{0,-\left|\sin\left(\frac{\kappa t}{2}\right)\right|\frac{t^{2}}{4n^{2}}\Vert A^{2}\varphi\Vert\right\}=0.
\end{equation*}
For the more interesting case of $[A,B]\neq0$, the bound can be made non-trivial by choosing $n$ large enough, as long as $\sin(\kappa t/2)\neq0$.
More explicitly, the bound is non-trivial for all
\begin{equation}
n\ge\frac{t}{4|\langle\psi | A\varphi\rangle|}\left[
\Vert A^{2}\varphi\Vert
+\frac{t}{6|{\sin(\kappa t/2)}|}
\,\Bigl(
\Vert[A,[A,B]]\Vert
+2\Vert[B,[B,A]]\Vert
\Bigr)
\right].
\label{eq:state_bound_non-trivial}
\end{equation}

\hypertarget{app:norm-bounds}{}
\section{C. Lower Bound on the Norm Trotter Error $b_n(t)$}
In this section, we prove a tighter version of the lower bound on the norm Trotter error $b_n(t)$.
That is, the bound
\begin{align}
b_n(t)
={}&\Vert T_{n}(t)-T(t)\Vert_\op \nonumber\\
\ge{}&
\max\!\left\{
\frac{t^2}{2n}
\left(
1
-
\frac{t}{2n}\|A\|_\op
\right)
Z_1(t)
,
\frac{t^2}{2n}
\left(
Z_1(t)
-\frac{t}{2n}
\|A\|_\op
Z_2(t)
\right)
\right\}
-\frac{t^{3}}{24n^{2}}\Vert[A,[A,B]]\Vert_\op
-\frac{t^{3}}{12n^{2}}\Vert[B,[B,A]]\Vert_\op,
\label{eqn:complete_norm_bound}
\end{align}
valid for $0\le t\le2n/\|A\|_\op$, where
\begin{align*}
Z_1(t)
&=\|[A,B]\|_\op
-t
\min\!\left\{
\frac{1}{2}\|[A+B,[A,B]]\|_\op
,
\|A\|_\op\|(A+B)^2\|_\op
\right\},\\
Z_2(t)
&=\frac{\|[A,[A,B]]\|_\op}{2\|A\|_\op}
+t\|A\|_\op\|(A+B)^2\|_\op.
\end{align*}
Note that
\begin{equation*}
\|[A,B]\|_\op\ge\frac{\|[A,[A,B]]\|_\op}{2\|A\|_\op},\qquad
\frac{1}{2}\|[A+B,[A,B]]\|_\op\le2\|A\|_\op\|A+B\|_\op^2.
\end{equation*}
Thus the bound presented in Eq.~(\ref{eqn:NormBound}) of the main text is a loose version,
\begin{equation}
b_n(t)
\ge
\frac{t^2}{2n}
\left(
1-\frac{t}{2n}\|A\|_\op
\right)
\left(
\|[A,B]\|_\op
-
\frac{t}{2}\|[A+B,[A,B]]\|_\op
\right)
-\frac{t^{3}}{24n^{2}}\Vert[A,[A,B]]\Vert_\op
-\frac{t^{3}}{12n^{2}}\Vert[B,[B,A]]\Vert_\op.
\label{eqn:complete_norm_bound_comm}
\end{equation}
We state this bound in the main text, since it actually maximizes the bound in Eq.~\eqref{eqn:complete_norm_bound} for the $XX$ spin chain we examine, and we expect it to be the best bound for most spin models.

\begin{proof}[Proof of the Lower Bound on the Norm Trotter Error $b_n(t)$]
We start with Eq.~\eqref{eq:proof_norm_bounds_start}.
The first term can be  bounded from below as
\begin{align}
\left\|
\rme^{-\rmi\frac{t}{2n}A}T(t)\rme^{\rmi\frac{t}{2n}A}-T(t)
\right\|_\op
&=
\left\|
\left(\rme^{-\rmi\frac{t}{2n}\ad_A}-\id\right)(T(t))
\right\|_\op
\nonumber\\
&=
\left\|
\left(\rme^{-\rmi\frac{t}{2n}\ad_A}-\id+\rmi\frac{t}{2n}\ad_A-\rmi\frac{t}{2n}\ad_A\right)(T(t))
\right\|_\op
\nonumber\\
&\ge
\frac{t}{2n}\|{\ad_A(T(t))}\|_\op
-
\left\|
\left(\rme^{-\rmi\frac{t}{2n}\ad_A}-\id+\rmi\frac{t}{2n}\ad_A\right)(T(t))
\right\|_\op
\nonumber\\
&\ge
\frac{t}{2n}\|{\ad_A(T(t))}\|_\op
-
\frac{t^2}{8n^2}
\|
{\ad_A^2(T(t))}
\|_\op
\nonumber\\
&=
\frac{t}{2n}\|[A,T(t)]\|_\op
-
\frac{t^2}{8n^2}
\|
[A,[A,T(t)]]
\|_\op,
\label{eq:proof_norm_bounds_rti}
\end{align}
where the reverse triangle inequality and a consequence of Lemma~\ref{lem:basic_lemma} are used in the third and fourth steps, respectively.

Let us bound $\|[A,T(t)]\|_\op$ from below.
We here provide two bounds.
The first one is obtained as
\begingroup
\allowdisplaybreaks
\begin{align*}
\|[A,T(t)]\|_\op
&=
\|
A\rme^{-\rmi t(A+B)}
-
\rme^{-\rmi t(A+B)}A
\|_\op
\nonumber\\
&=
\|
\rme^{\rmi t(A+B)}A\rme^{-\rmi t(A+B)}
-
A
\|_\op
\nonumber\\
&=
\|
(
\rme^{\rmi t\ad_{A+B}}
-\id
)
(A)
\|_\op
\nonumber\\
&=
\|
(
\rme^{\rmi t\ad_{A+B}}
-\id
-\rmi t\ad_{A+B}
+\rmi t\ad_{A+B}
)
(A)
\|_\op
\nonumber\\
&\ge
t\|{\ad_{A+B}(A)}\|_\op
-
\|
(
\rme^{\rmi t\ad_{A+B}}
-\id
-\rmi t\ad_{A+B}
)
(A)
\|_\op
\nonumber\\
&\ge
t\|{\ad_{A+B}(A)}\|_\op
-
\frac{t^2}{2}
\|
{\ad_{A+B}^2(A)}
\|_\op
\nonumber\\
&=
t\|[A,B]\|_\op
-
\frac{t^2}{2}
\|
[A+B,[A,B]]
\|_\op,
\intertext{while another one is obtained as}
\|[A,T(t)]\|_\op
&=\|[A,\rme^{-\rmi t(A+B)}]\|_\op
\nonumber\\
&=\|[A,\rme^{-\rmi t(A+B)}-I+\rmi t(A+B)-\rmi t(A+B)]\|_\op
\nonumber\\
&\ge
t\|[A,A+B]\|_\op-\|[A,\rme^{-\rmi t(A+B)}-I+\rmi t(A+B)]\|_\op
\nonumber\\
&\ge
t\|[A,B]\|_\op
-t^2\|A\|_\op\|(A+B)^2\|_\op.
\intertext{Combining these bounds, we get}
	\Vert [A,T(t)]\Vert_\op
	&\geq t\Vert[A,B]\Vert_\op
	-t^2\min\!\left\{
	\frac{1}{2}\|[A+B,[A,B]]\|_\op,
	\Vert A\Vert_\op\Vert(A+B)^2\Vert_\op
	\right\}
	=tZ_1(t).
\end{align*}
\endgroup

Let us next bound $\|[A,[A,T(t)]]\|_\op$ from above.
We again take two different strategies.
The first one proceeds as
\begin{align*}
\|[A,[A,T(t)]]\|_\op
&=\|[A,[A,\rme^{-\rmi t(A+B)}]]\|_\op
\nonumber\\
&=\|A^2\rme^{-\rmi t(A+B)}
-2A\rme^{-\rmi t(A+B)}A
+\rme^{-\rmi t(A+B)}A^2\|_\op
\nonumber\\
&=
\|
\rme^{\rmi t(A+B)}A^2\rme^{-\rmi t(A+B)}
-2\rme^{\rmi t(A+B)}A\rme^{-\rmi t(A+B)}A
+A^2
\|_\op
\nonumber\\
&=
\Bigl\|
\rme^{\rmi t(A+B)}A\rme^{-\rmi t(A+B)}
(\rme^{\rmi t(A+B)}A\rme^{-\rmi t(A+B)}-A)
-
(\rme^{\rmi t(A+B)}A\rme^{-\rmi t(A+B)}-A)A
\Bigr\|_\op
\nonumber\\
&=
\Bigl\|
\rme^{\rmi t\ad_{A+B}}(A)
(\rme^{\rmi t\ad_{A+B}}-\id)(A)
-
(\rme^{\rmi t\ad_{A+B}}-\id)(A)
A
\Bigr\|_\op
\nonumber\\
&\le 
2\|A\|_\op
\|
(\rme^{\rmi t\ad_{A+B}}-\id)(A)
\|_\op
\nonumber\\
&\le 
2t\|A\|_\op
\|
{\ad_{A+B}(A)}
\|_\op
\nonumber\\
&
=
2t\|A\|_\op
\|
[A,B]
\|_\op,
\intertext{where the inequality in the penultimate step follows from Lemma~\ref{lem:basic_lemma}\@. The other strategy proceeds as}
\|[A,[A,T(t)]]\|_\op
&=\|[A,[A,\rme^{-\rmi t(A+B)}]]\|_\op
\nonumber\\
&=\|[A,[A,\rme^{-\rmi t(A+B)}-I+\rmi t(A+B)-\rmi t(A+B)]]\|_\op
\nonumber\\
&\le
t\|[A,[A,A+B]]\|_\op
+\|[A,[A,\rme^{-\rmi t(A+B)}-I+\rmi t(A+B)]]\|_\op
\nonumber\\
&\le
t\|[A,[A,B]]\|_\op
+4\|A\|_\op^2\|\rme^{-\rmi t(A+B)}-I+\rmi t(A+B)\|_\op
\nonumber\\
&\le
t\|[A,[A,B]]\|_\op
+2t^2\|A\|_\op^2\|(A+B)^2\|_\op,
\intertext{where the last step again invokes Lemma~\ref{lem:basic_lemma}\@. Thus,}
\|[A,[A,T(t)]]\|_\op
&\le 
2t\min\!\left\{
\|A\|_\op\|[A,B]\|_\op
,
\frac{1}{2}
\|[A,[A,B]]\|_\op
+t\|A\|_\op^2\|(A+B)^2\|_\op
\right\}
\\
&=2t\|A\|_\op
\min\Bigl\{
\|[A,B]\|_\op
,
Z_2(t)
\Bigr\}.
\end{align*}

Using these bounds, we can bound Eq.~(\ref{eq:proof_norm_bounds_rti}) as
\begin{equation*}
\left\|
\rme^{-\rmi\frac{t}{2n}A}T(t)\rme^{\rmi\frac{t}{2n}A}-T(t)
\right\|_\op
\ge
\frac{t^2}{2n}
Z_1(t)
-\frac{t^3}{4n^2}
\|A\|_\op
\min\Bigl\{
\|[A,B]\|_\op
,
Z_2(t)
\Bigr\}.
\end{equation*}
Or, via
\begin{equation*}
\left\Vert \left[A,\left[A,T(t)\right]\right]\right\Vert_\op
	\leq \left\Vert A\right\Vert_\op\left\Vert\left[A,T(t)\right]\right\Vert_\op,
\end{equation*}
we can alternatively bound Eq.~(\ref{eq:proof_norm_bounds_rti}) as
\begin{align*}
\left\|
\rme^{-\rmi\frac{t}{2n}A}T(t)\rme^{\rmi\frac{t}{2n}A}-T(t)
\right\|_\op
&\ge
\frac{t}{2n}
\left(
1-\frac{t}{2n}\|A\|_\op
\right)
\|[A,T(t)]\|_\op
\nonumber\\
&
\ge
\frac{t^2}{2n}
\left(
1-\frac{t}{2n}\|A\|_\op
\right)
Z_1(t).
\end{align*}
This last inequality is valid for $t\le2n/\|A\|_\op$.
These bounds are then combined as
\begingroup
\allowdisplaybreaks
\begin{align*}
\left\|
\rme^{-\rmi\frac{t}{2n}A}T(t)\rme^{\rmi\frac{t}{2n}A}-T(t)
\right\|_\op
&\ge
\max\!\left\{
\frac{t^2}{2n}
\left(
1
-
\frac{t}{2n}\|A\|_\op
\right)
Z_1(t)
,
\frac{t^2}{2n}
\left(
Z_1(t)
-\frac{t}{2n}
\|A\|_\op
Z_2(t)
\right)
\right\}.
\end{align*}
\endgroup
Note that we have taken into account the fact $Z_1(t)\le\|[A,B]\|_\op$. 
A bound on the second term of Eq.~\eqref{eq:proof_norm_bounds_start} is given by Eq.~(\ref{eqn:Childs2021}), and we get the bound~(\ref{eqn:complete_norm_bound}).
\end{proof}

If $[A,B]=0$, we have $b_n(t)=0$, while the bounds~(\ref{eqn:complete_norm_bound}) and~(\ref{eqn:complete_norm_bound_comm}) both reduce to $b_n(t)\ge0$.
For $[A,B]\neq0$, the bound becomes non-trivial for $t$ small enough.
More explicitly, on the basis of the bound~(\ref{eqn:complete_norm_bound_comm}), it suffices to take
\begin{equation}
t\le\frac{
12n\|[A,B]\|_\op
}{
6n\|[A+B,[A,B]]\|_\op
+6\|A\|_\op\|[A,B]\|_\op
+\Vert[A,[A,B]]\Vert_\op
+2\Vert[B,[B,A]]\Vert_\op
},
\label{eq:norm_non-trivial}
\end{equation}
which is compatible with the valid time range $t\le2n/\|A\|$ for the bound with moderate $n$.

Since $b_n(t)\ge\xi_n(t;\varphi)$, the lower bound on $\xi_n(t;\varphi)$ can be useful also for $b_n(t)$. 
A refined version of the lower bound on $b_n(t)$ is then given by comparing the lower bounds on $b_n(t)$ and on $\xi_n(t;\varphi)$ as 
\begin{align*}
b_n(t)\ge \max\bigg\{0,
\frac{t^2}{2n}
\left(
1
-
\frac{t}{2n}\|A\|_\op
\right)
Z_1(t)
,
\frac{t^2}{2n}
\left(
Z_1(t)
-\frac{t}{2n}
\|A\|_\op
Z_2(t)
\right)
,
\left|\sin\left(\frac{\kappa t}{2}\right)\right|\left(\frac{t}{n}|\langle \psi | A\varphi\rangle|
  -\frac{t^{2}}{4n^{2}}\Vert A(g)^{2}\varphi\Vert \right) 
  \biggr\}
\hphantom{.}\\
-\frac{t^{3}}{24n^{2}}\Vert[A,[A,B]]\Vert_\op
  -\frac{t^{3}}{12n^{2}}\Vert[B,[B,A]]\Vert_\op
  \bigg\}.
\end{align*}

\hypertarget{app:examples}{}
\section{D. Examples}
In this appendix, we study our bounds on the Trotter errors with two physical examples.
The first one is a simple single-qubit model and the other one is the $XX$ model.
Throughout, the Pauli matrices are denoted by $X$, $Y$, and $Z$, i.e.,
\begin{equation*}
X=\begin{pmatrix}
\medskip
0&1\\
1&0
\end{pmatrix},
\qquad
Y=\begin{pmatrix}
\medskip
0&-\rmi\\
\rmi&0
\end{pmatrix},
\qquad
Z=\begin{pmatrix}
\medskip
1&0\\
0&-1
\end{pmatrix}.
\end{equation*}

\hypertarget{app:1-qubit_example}{}
\subsection{D.1. Single-Qubit Example}
Let us consider the case with $A=X$ and $B=Z$.
In this case, we have
\begin{align*}
(A+B)\varphi&=-\sqrt{2}\,\varphi,\quad
\varphi=\frac{1}{\sqrt{2}}
\begin{pmatrix}
\medskip
-\sqrt{1-1/\sqrt{2}}
\\
\sqrt{1+1/\sqrt{2}}
\end{pmatrix},\\
(A+B)\psi&=\sqrt{2}\,\psi,\quad
\psi=\frac{1}{\sqrt{2}}
\begin{pmatrix}
\medskip
\sqrt{1+1/\sqrt{2}}
\\
\sqrt{1-1/\sqrt{2}}
\end{pmatrix},
\end{align*}
and we set $h=-\sqrt{2}$.
The spectral gap is given by $\lambda=2\sqrt{2}$.
Then, for $A(g)=A-h+g$ and $B(g)=B-g$, we have
\[
|\langle \psi | A\varphi\rangle|=\frac{1}{\sqrt{2}},
\]
and
\[
\Vert A(g)^{2}\varphi\Vert=\sqrt{5+6\sqrt{2}\,g+6g^2+2\sqrt{2}\,g^3+g^4},\qquad
\Vert B(g)^{2}\varphi\Vert=\sqrt{1+2\sqrt{2}\,g+6g^2+2\sqrt{2}\,g^3+g^4}.
\]
Furthermore, we have
\begin{equation*}
\Vert [A,B]\Vert_\op =2,\qquad
\Vert[A,[A,B]]\Vert_\op
=\Vert[B,[B,A]]\Vert_\op
=4,\qquad
\Vert[A+B,[B,A]]\Vert_\op
=4\sqrt{2},
\end{equation*}
for the commutators, and
\begin{equation*}
	\Vert A\Vert_\op = 1,\qquad
	\Vert(A+B)^2\Vert_\op=2.
\end{equation*}

Since $a=\min_g\Vert A(g)^{2}\varphi\Vert=\frac{1}{2}\sqrt{6(\sqrt{2}+1)^{2/3}-3(\sqrt{2}+1)^{4/3}+6(\sqrt{2}-1)^{2/3}-3(\sqrt{2}-1)^{4/3}-1}=0.788903$ at $g=-\frac{1}{\sqrt{2}}\left(\sqrt[3]{\sqrt{2}+1}-\sqrt[3]{\sqrt{2}-1}+1\right)=-1.12859$, the lower bound on the state-dependent Trotter error $\xi_n(t;\varphi)$ [Eq.~(\ref{eqn:StateBound}) of the main text] is optimized to
\begin{equation*}
\xi_n(t;\varphi)\geq \max\!\left\{0,
\frac{t^2}{n}
\left(
1-\frac{at}{2\sqrt{2}\,n}
\right)
\frac{|{\sin\sqrt{2}\,t}|}{\sqrt{2}\,t}
-\frac{t^{3}}{2n^{2}}
\right\}.
\end{equation*}
This bound is non-trivial for all
\begin{equation*}
n
\ge
\frac{at}{2\sqrt{2}}
\left(
1
+
\frac{2t}{a|{\sin\sqrt{2}\,t}|}
\right).
\end{equation*}
On the other hand, since $\min_g(\Vert A(g)^{2}\varphi\Vert+\Vert B(g)^{2}\varphi\Vert)=\sqrt{5}$ at $g=-1/\sqrt{2}$, the upper bound on the state-dependent Trotter error $\xi_n(t;\varphi)$ [Eq.~(\ref{eq:upper_state_bound}) of the main text] is optimized to
\begin{equation*}
\xi_n(t;\varphi)\leq \frac{\sqrt{5}\,t^2}{2n}.
\end{equation*}

For the norm Trotter error $b_n(t)$, the lower bound~(\ref{eqn:complete_norm_bound}) yields
\begin{equation*}
b_n(t)\ge 
\max\!\left\{
0,
\frac{t^2}{n}
\left(
1-\frac{t}{2n}
\right)
(1-t)
-\frac{t^3}{2n^2}
\right\},
\end{equation*}
while the upper bound in Eq.~(\ref{eq:upper_norm_bound}) of the main text gives
\begin{equation*}
b_n(t)\le 
\frac{t^2}{n}.
\end{equation*}
The lower bound is non-trivial for 
\[
t\le n+1-\sqrt{n^2+1}.
\]
It suffices to take $t\le1-1/(2n)<n+1-\sqrt{n^2+1}$.

Since $\xi_n(t;\varphi)\le b_n(t)$, the upper bound on $b_n(t)$ is also an upper bound on $\xi_n(t;\varphi)$, and the lower bound on $\xi_n(t;\varphi)$ is also a lower bound on $b_n(t)$.
Putting these together, we get
\begin{mdframed}
\begin{equation}
\max\!\left\{0,
\frac{t^2}{n}
\left(
1-\frac{at}{2\sqrt{2}\,n}
\right)
\frac{|{\sin\sqrt{2}\,t}|}{\sqrt{2}\,t}
-\frac{t^{3}}{2n^{2}}
\right\}
\le\xi_n(t;\varphi)\le b_n(t)\le\frac{t^2}{n}.
\label{eq:1qubit_state_norm_bound}
\end{equation}
\end{mdframed}

\hypertarget{app:XX_example}{}
\subsection{D.2. The $XX$ Model}
Here, we test our bounds with the isotropic one-dimensional $XX$ model, introduced in Ref.~\cite{Lieb1961}\@.
The $XX$ model consists of a chain of $L$ spin-$1/2$ spins with the nearest neighbor interactions.
Its Hamiltonian reads
\begin{equation}
    H = \frac{J}{4} \sum_{j=0}^{L-1} ( X_j X_{j+1} +  Y_j Y_{j+1} ),
    \label{eq:XX_Hamiltonian}
\end{equation}
where $X_j$, $Y_j$, and $Z_j$ denote the actions of the first, second, and third Pauli operators on site $j$, and $J\in\mathbb{R}$ is the coupling constant.
In the following, we will assume that $J>0$.
We impose the periodic boundary conditions, $X_L=X_0$, $Y_L=Y_0$, and $Z_L=Z_0$. 
We consider the Trotter product formula with
\begin{equation*}
A = \frac{J}{4} \sum_{j=0}^{L-1}  X_j X_{j+1} 
\qquad\text{and}\qquad
B = \frac{J}{4} \sum_{j=0}^{L-1}  Y_j Y_{j+1}.
\end{equation*}
Evidently, $A$ and $B$ can be diagonalized trivially by going to the basis of the respective Pauli operators.
The Hamiltonian $A+B$ of the $XX$ model~\eqref{eq:XX_Hamiltonian} can also be diagonalized analytically~\cite{DePasquale2008,DePasquale2009,Franchini2017}\@.
This allows us to analytically compute the upper and lower bounds on the state-dependent Trotter error $\xi_n(t;\varphi)$.

\subsubsection{Diagonalization of the $XX$ Model}
We here recall the exact diagonalization of the $XX$ model~(\ref{eq:XX_Hamiltonian})~\cite{DePasquale2008,DePasquale2009,Franchini2017}, which can be rewritten as
\begin{equation}
	 H=\frac{J}{2} \sum_{j=0}^{L-1} (\sigma^+_j \sigma^-_{j+1} + \sigma^-_j \sigma^+_{j+1}),
	 \label{eq:XX_Hamil_creation_annihilation}
\end{equation}
where $\sigma_j^\pm=\frac{1}{2}(X_j\pm\rmi Y_j)$ are the ladder operators of the $j$th spin.
We perform a Jordan-Wigner transformation to express \eqref{eq:XX_Hamil_creation_annihilation} by spinless fermions $c_j$.
To this end, we replace $\sigma_j^{\pm}$ with
\begin{equation*}
    \sigma_j^- = \rme^{\pi\rmi n_{j\uparrow}} c_j, \qquad
    \sigma_j^+ = \rme^{\pi\rmi n_{j\uparrow}} c_j^\dagger, \qquad
    n_{j\uparrow}=\sum_{i=0}^{j-1}\sigma_i^+\sigma_i^-=\sum_{i=0}^{j-1}c_i^\dag c_i,\qquad
    j\in\{0,\ldots,L-1\}.
\end{equation*}
The new operators $c_j$ and $c_j^\dag$ satisfy the fermionic anticommutation relations,
\begin{equation*}
\{c_i,c_j\}=0,\qquad
\{c_i^\dag,c_j^\dag\}=0,\qquad
\{c_i,c_j^\dag\}=\delta_{i,j},\qquad
    i,j\in\{0,\ldots,L-1\}.
\end{equation*}
In terms of these fermionic operators, the Hamiltonian reads
\begin{equation}
H=\frac{J}{2}\left(
\sum_{j=0}^{L-2}(c_j^\dag c_{j+1}+c_{j+1}^\dag c_j)
-\rme^{\pi\rmi n_\uparrow}(c_{L-1}^\dag c_0+c_0^\dag c_{L-1})
\right),
\label{eq:XX_Hamil_fermionic}
\end{equation}
where $n_\uparrow=n_{L\uparrow}$ is the number of up spins in the chain.
Since 
\begin{equation*}
[\rme^{\pi\rmi n_\uparrow},H]=0,
\end{equation*}
the Hamiltonian $H$ is diagonalizable for $\rme^{\pi\rmi n_\uparrow}=1$ and for $\rme^{\pi\rmi n_\uparrow}=-1$ separately.
Depending on the parity of $n_\uparrow$, we perform a (deformed) Fourier transformation
\begin{equation*}
c_j
=\begin{cases}
\medskip
\displaystyle
\frac{1}{\sqrt{L}}\sum_{k=0}^{L-1}\rme^{2\pi\rmi jk/L}\hat{c}_k^{(-)}
&(n_\uparrow\ \text{odd}),\\
\displaystyle
\frac{1}{\sqrt{L}}\sum_{k=0}^{L-1}\rme^{2\pi\rmi j(k+1/2)/L}\hat{c}_k^{(+)}
&(n_\uparrow\ \text{even}),
\end{cases}
\qquad
j\in\{0,\ldots,L-1\}.
\end{equation*}
This preserves the fermionic anticommutation relations,
\begin{equation*}
\{\hat{c}_k^{(\pm)},\hat{c}_\ell^{(\pm)}\}=0,\qquad
\{\hat{c}_k^{(\pm)\dag},\hat{c}_\ell^{(\pm)\dag}\}=0,\qquad
\{\hat{c}_k^{(\pm)},\hat{c}_\ell^{(\pm)\dag}\}=\delta_{k,\ell},\qquad
k,\ell\in\{0,\ldots,L-1\}.
\end{equation*}
By this transformation, the Hamiltonian~(\ref{eq:XX_Hamil_fermionic}) is diagonalized as
\begin{equation*}
H=J
\sum_{k=0}^{N-1}
\left[
\cos\!\left(\frac{2\pi k}{N}\right)
\hat{c}_k^{(-)\dag}\hat{c}_k^{(-)}
\frac{1-\rme^{\pi\rmi n_\uparrow}}{2}
+
\cos\!\left(\frac{2\pi(k+1/2)}{N}\right)
\hat{c}_k^{(+)\dag}\hat{c}_k^{(+)}
\frac{1+\rme^{\pi\rmi n_\uparrow}}{2}
\right].
\end{equation*}
The eigenvalues corresponding to $\ket{\downarrow \cdots \downarrow}$ and $\ket{\uparrow \cdots \uparrow}$ are zero, which can be easily seen by acting the Hamiltonian $H$ in the representation~\eqref{eq:XX_Hamil_creation_annihilation}\@.

\subsubsection{Bound on the State-Dependent Trotter Error $\xi_n(t;\varphi)$}
We look at the Trotter error $\xi_n(t;\varphi)$ on the eigenstate $\varphi=\ket{\downarrow \cdots \downarrow}$ with no spin excitations.
To obtain a non-trivial lower bound on $\xi_n(t;\varphi)$, we have to choose an eigenstate $\psi$ belonging to a non-zero eigenvalue, such that the term $\langle\psi,A\varphi\rangle$ is non-vanishing.
As $A$ flips two spins simultaneously, we consider states with two excitations, 
\begin{equation*}
    \Psi(k_1,k_2)
    = \hat{c}_{k_1}^{(+)\dag}\hat{c}_{k_2}^{(+)\dag} {\ket{\downarrow  \cdots \downarrow}}.
\end{equation*}
Note that $k_1 = k_2$ yields zero, as the excitations are fermionic, so that $(\hat{c}_k^\dag)^2 = 0$.
The corresponding eigenvalues are
\begin{equation*}
H{{\Psi(k_1,k_2)}} 
= J \left[\cos\!\left(\frac{2\pi(k_1+1/2)}{L}\right) + \cos\!\left(\frac{2\pi(k_2+1/2)}{L}\right)\right] {{\Psi(k_1,k_2)}}.
\end{equation*}
For the term $\langle\psi | A\varphi\rangle$ with $\varphi=\ket{\downarrow \cdots \downarrow}$ and $\psi={\Psi(k_1,k_2)}$, we obtain
\begin{align*}
\langle\Psi(k_1,k_2) | A\varphi\rangle
    &=  \frac{J}{4} \sum_{j=0}^{L-1} {\bra{\Psi(k_1,k_2)}}X_j X_{j+1}{\ket{\downarrow\cdots\downarrow}}\\
    &=  \frac{J}{4} \sum_{j=0}^{L-1} {\bra{\Psi(k_1,k_2)}}\sigma_j^+\sigma_{j+1}^+{\ket{\downarrow\cdots\downarrow}}\\
    &=
    \frac{J}{4} \sum_{j=0}^{L-2}
    {\bra{\Psi(k_1,k_2)}}
    c_j^\dag c_{j+1}^\dag 
    {\ket{\downarrow\cdots\downarrow}}
    -\frac{J}{4}
    {\bra{\Psi(k_1,k_2)}}
    c_{L-1}^\dag
    c_0^\dag
    {\ket{\downarrow\cdots\downarrow}}\\
    &=  \frac{J}{4L} \sum_{j=0}^{L-1}\sum_{\ell_1=0}^{L-1}\sum_{\ell_2=0}^{L-1}
    \rme^{-2\pi\rmi j(\ell_1+1/2)/L}
    \rme^{-2\pi\rmi(j+1)(\ell_2+1/2)/L}
     {\bra{\Psi(k_1,k_2)}}
     \hat{c}_{\ell_1}^{(+)\dag}\hat{c}_{\ell_2}^{(+)\dag}{\ket{\downarrow\cdots\downarrow}}\\
&=-\frac{\rmi J}{2}
    \sin\!\left(\frac{2\pi(k_1+1/2)}{L}\right)
    \delta_{k_1+k_2,L-1}.
\end{align*}
This is non-vanishing only for $k_1 + k_2 = L-1$.
For these states, the spectral gap $\lambda$ is given by
\begin{align*}
    \lambda &= \langle\Psi(k,L-1-k) | H\Psi(k,L-1-k)\rangle - \langle\varphi,H\varphi\rangle \\
    &= 2J\cos\!\left(\frac{2\pi(k+1/2)}{L}\right).
\end{align*}

Next, we compute $\|A(g)^2\varphi\|$.
We set $g=0$, so that $\|A(g)^2 \varphi\|=\|A^2 \varphi\|$.
Assuming $L\ge5$, we get
\begin{align*}
\|A^2\varphi\|^2 
&= \left(\frac{J}{4}\right)^4
\sum_{j_1 = 0}^{L-1} 
\sum_{j_2 = 0}^{L-1} 
\sum_{j_3 = 0}^{L-1} 
\sum_{j_4 = 0}^{L-1} 
{\bra{\downarrow\cdots\downarrow}}
X_{j_1}X_{j_1+1}X_{j_2}X_{j_2+1}X_{j_3}X_{j_3+1}X_{j_4}X_{j_4+1}
{\ket{\downarrow\cdots\downarrow}}\\
&= \frac{J^4}{4^4}
\left(
\sum_{j_1 = 0}^{L-1} 
\sum_{j_2 = 0}^{L-1} 
\sum_{j_3 = 0}^{L-1} 
\sum_{j_4 = 0}^{L-1} 
(\delta_{j_1,j_2}\delta_{j_3,j_4}
+\delta_{j_1,j_3}\delta_{j_2,j_4} 
+\delta_{j_1,j_4}\delta_{j_1,j_3}
)-2L
\right) \\
&=\frac{J^4}{4^4}L(3L-2),
\end{align*}
where the second equality follows from the fact that the states with different spins are orthogonal.
In turn, we always need an even number of $X$ operators acting on the same site to get a non-zero result.
The term $-2L$ corrects the overcounting of the terms where all four indices are the same, $j_1=j_2=j_3=j_4$.
Note that for $L=4$ there are additional terms of the form $\delta_{j_1+1,j_2}\delta_{j_2+1,j_3}\delta_{j_3+1,j_4}$ contributing.

To compute the norms of the double-commutators in the bound, we first compute the commutator $[A,B]$,
\begin{align}
[A,B]
&=\frac{J^2}{16}
\sum_{i=0}^{L-1}
\sum_{j=0}^{L-1}
[X_i X_{i+1},Y_j Y_{j+1}]
\nonumber\\
&=
\frac{\rmi J^2}{8}
\sum_{j=0}^{L-1}
(
X_jY_{j+2}
+
Y_jX_{j+2}
)
Z_{j+1}.
\label{eq:commutator}
\end{align}
This allows us to compute the double-commutators
\begin{align*}
[A,[A,B]]
&=
\frac{\rmi J^3}{32}
\sum_{i=0}^{L-1}
\sum_{j=0}^{L-1}
[
X_iX_{i+1}
,
(
X_jY_{j+2}
+
Y_jX_{j+2}
)
Z_{j+1}
]
\nonumber\\
&=
\frac{J^3}{8}
\sum_{j=0}^{L-1}
(
Y_j
Y_{j+1}
-
X_j
Z_{j+1}
Z_{j+2}
X_{j+3}
),
\end{align*}
and similarly
\begingroup
\allowdisplaybreaks
\begin{align*}
[B,[B,A]]
&=-\frac{\rmi J^3}{32}
\sum_{i=0}^{L-1}
\sum_{j=0}^{L-1}
[Y_iY_{i+1},(X_j Y_{j+2} +  Y_jX_{j+2}) Z_{j+1}]\\
&=\frac{J^3}{8}
\sum_{j=0}^{L-1}
(
X_jX_{j+1}
-
Y_jZ_{j+1}Z_{j+2}Y_{j+3}
).
\end{align*}
\endgroup
We bound the norms of the double-commutators $\|[A,[A,B]]\|_\op$ and $\|[B,[B,A]]\|_\op$ by using the triangle inequality as
\begin{align*}
\|[A,[A,B]]\|_\op 
&=\frac{J^3}{8} \left\|\sum_{j=0}^{L-1} (Y_jY_{j+1} - X_jZ_{j+1}Z_{j+2}X_{j+3})\right\|_\op \\
&\leq \frac{J^3}{8} \sum_{j=0}^{L-1}\Bigl(
\|Y_jY_{j+1}\|_\op 
+ \|X_jZ_{j+1}Z_{j+2}X_{j+3}\|_\op
\Bigr)\\
&=\frac{J^3}{8} \sum_{j=0}^{L-1} \Bigl(
\|Y\|_\op^2 + \|X\|_\op^2\|Z\|_\op^2
\Bigr)\\
&
= \frac{J^3}{4} L,
\intertext{and}
\|[B,[B,A]]\|_\op 
&\le \frac{J^3}{4} L.
\end{align*}

Putting these elements together, the state-dependent Trotter error $\xi_n(t;\varphi)$ is bounded (for $L \geq 5$) by
\begin{mdframed}
\begin{multline}
\max\!\left\{0,\left|
\sin\!\left[
Jt\cos\!\left(\frac{2\pi(k+1/2)}{L}\right)
\right]
\right|
\left(
\frac{Jt}{2n}
\left|
    \sin\!\left(\frac{2\pi(k+1/2)}{L}\right)
\right|
-
\frac{J^2t^2}{64n^2}\sqrt{L(3L-2)}
\right) 
-\frac{J^3t^{3}}{32n^{2}}L
\right\}
\\
\le\xi_n(t;\varphi)
\le\frac{J^2t^2}{16n}\sqrt{L(3L-2)}.
\label{eq:chain_state_bound}
\end{multline}
\end{mdframed}
To obtain the best possible bound, we can take the supremum over $k\in\{0,1,\ldots,L-1\}$.
This lower bound is non-trivial for
\begin{equation*}
n
\ge
\frac{JtL}{
32|{\sin[2\pi(k+1/2)/L]}|
}
\left(
\sqrt{3-2/L}
+
\frac{
2Jt
}{
|
{
\sin\{
Jt\cos[2\pi(k+1/2)/L]
\}
}
|
}
\right).
\end{equation*}

\subsubsection{Bound on the Norm Trotter Error $b_n(t)$}
Many of the quantities involved in the bounds on the norm Trotter error $b_n(t)$ have already been computed for the bounds on the state-dependent Trotter error $\xi_n(t;\varphi)$.
For instance, the commutator $[A,B]$ is calculated in Eq.~\eqref{eq:commutator}\@.
Its operator norm $\Vert[A,B]\Vert_\op$ is upper bounded by
\begin{align*}
\|[A,B]\|_\op
&=
\frac{J^2}{8}
\left\|
\sum_{j=0}^{L-1}
(
X_jY_{j+2}
+
Y_jX_{j+2}
)
Z_{j+1}
\right\|_\op
\nonumber\\
&\le 
\frac{J^2}{8}
\sum_{j=0}^{L-1}
\Bigl(
\|X_jZ_{j+1}Y_{j+2}\|_\op
+
\|Y_jZ_{j+1}X_{j+2}\|_\op
\Bigr)
\nonumber\\
&=
\frac{J^2}{4}L.
\end{align*}
We also need its lower bound. 
To bound it from below, we use the equivalence of the matrix norms.
That is, for all $M\in\mathbb{C}^{d\times d}$, we have
\begin{equation}
	\Vert M\Vert_\op \geq \frac{1}{\sqrt{d}} \Vert M\Vert_\F,
\end{equation}
where $\Vert M\Vert_\F=\sqrt{\tr(M^\dag M)}$ is the Frobenius norm~\cite[Eq.~(1.169)]{Watrous2018}.
Applying this to Eq.~\eqref{eq:commutator}, we get
\begingroup
\allowdisplaybreaks
\begin{align*}
\|[A,B]\|_\op^2
&\geq
\frac{1}{2^L}
\|[A,B]\|_\F^2
\nonumber\\
&=
\frac{1}{2^L}
\frac{J^4}{64}
\tr\!\left[
\left(
\sum_{j=0}^{L-1}
(
X_jY_{j+2}
+
Y_jX_{j+2}
)
Z_{j+1}
\right)^2
\right]
\nonumber\\
&=
\frac{1}{2^L}
\frac{J^4}{64}
\sum_{j=0}^{L-1}
\tr
[
(
X_jZ_{j+1}Y_{j+2}
)^2
+
(
Y_jZ_{j+1}X_{j+2}
)^2
]
\nonumber\\
&=
\frac{1}{2^L}
\frac{J^4}{64}
\sum_{j=0}^{L-1}
2^{L+1}
\nonumber\\
&=
\frac{J^4}{32}
L.
\end{align*}
\endgroup
Therefore,
\begin{equation}
\frac{J^2}{4} \sqrt{\frac{L}{2}}
\le\|[A,B]\|_\op
\le\frac{J^2}{4}L.
\label{eqn:UpperBoundCommAB}
\end{equation}
Additionally, we need upper bounds on $\|A\|_\op$ and $\|A+B\|_\op$.
To this end, we again use the triangle inequality to get
\begin{align*}
\|A\|_\op
&=
\frac{J}{4}
\left\|
\sum_{j=0}^{L-1}
X_jX_{j+1}
\right\|_\op
\nonumber\\
&\le 
\frac{J}{4}
\sum_{j=0}^{L-1}
\|X_jX_{j+1}\|_\op
\nonumber\\
&=
\frac{J}{4}L,
\intertext{and}
    \|A+B\|_\op
    &= \frac{J}{4} \left\|\sum_{j=1}^L (X_j X_{j+1} + Y_j Y_{j+1})\right\|_\op\\
    &\leq \frac{J}{4} \sum_{j=1}^L \Bigl(
    \|X_j X_{j+1}\|_\op + \|Y_j Y_{j+1}\|_\op \Bigr)\\
    &= \frac{J}{2} L .
\end{align*}
Using these bounds, we get
\begin{align*}
Z_1(t)
&=\|[A,B]\|_\op
-t
\min\!\left\{
\frac{1}{2}\|[A+B,[A,B]]\|_\op
,
\|A\|_\op\|(A+B)^2\|_\op
\right\}\\
&\ge 
\frac{J^2}{4} \sqrt{\frac{L}{2}}
-\min\!\left\{
\frac{J^3t}{4}L
,
\frac{J^3t}{16}L^3
\right\}
\\
&=
\frac{J^2}{4} \sqrt{\frac{L}{2}}
-\frac{J^3t}{4} L,
\end{align*}
for $L\ge2$, and
\begin{align*}
\|A\|_\op Z_2(t)
&=\frac{1}{2}\|[A,[A,B]]\|_\op
+t\|A\|_\op^2\|(A+B)^2\|_\op
\\
&\le 
\frac{J^3}{8}L
+\frac{J^4t}{64}L^4.
\end{align*}
These allow us to bound the lower bound in Eq.~(\ref{eqn:complete_norm_bound}) from below.
On the other hand, the upper bound in Eq.~(\ref{eq:upper_norm_bound}) of the main text can be bounded from above using the upper bound on $\|[A,B]\|_\op$ in Eq.~(\ref{eqn:UpperBoundCommAB}).
Overall, the norm Trotter error $b_n(t)$ is bounded (for $L \geq 5$ and $Jt\le8n/L$) by
\begin{mdframed}
\begin{align}
\max\biggl\{&
0,
\frac{J^2t^2}{8n}
\,\biggl(
\sqrt{\frac{L}{2}}
-JtL
\biggr)
-
\frac{J^3t^3}{64n^2}
L
\min\biggl\{
\sqrt{\frac{L}{2}}
-JtL
,
2+\frac{Jt}{4}L^3
\biggr\}
-\frac{J^3t^{3}}{32n^{2}}L
,
\nonumber\\
&
\left|
\sin\!\left[
Jt\cos\!\left(\frac{2\pi(k+1/2)}{L}\right)
\right]
\right|
\left(
\frac{Jt}{2n}
\left|
\sin\!\left(\frac{2\pi(k+1/2)}{L}\right)
\right|
-
\frac{J^2t^2}{64n^2}\sqrt{L(3L-2)}
\right) 
-\frac{J^3t^{3}}{32n^{2}}L
\biggr\}
\le
b_n(t)
\le
\frac{J^2t^2}{8n}L.
\label{eq:chain_norm_bound}
\end{align}
\end{mdframed}
Note that this upper bound on $b_n(t)$ is larger than the upper bound on $\xi_n(t;\varphi)$ in Eq.~(\ref{eq:chain_state_bound}).
\end{document}